\documentclass[11pt, twoside]{rmta2018eng}
\usepackage[T1]{fontenc} % True fonts
\usepackage{ae,aecompl} % Enhanced encoding for True Fonts 
\usepackage{times}
\usepackage{fancyhdr}
\usepackage{amsfonts}
\usepackage{amsmath}
\usepackage{amssymb}
\usepackage{cite}
\usepackage{soul,xcolor}
\usepackage[hang,small,bf]{caption}
\usepackage[latin1]{inputenc} % For spanish words (do not use Babel)
\usepackage{hyperref}\usepackage{breakurl} 
\usepackage{graphicx} %  Figure only in format encapsulated postscript (.eps)
\usepackage{float}
\usepackage{subfig} %este
\usepackage{cleveref}

\paginas{1}{10} % Page numbers, the Editor puts the final numbers;

\newtheorem{theorem}{Theorem}[section]
\newtheorem{example}[theorem]{Example}
\newtheorem{lemma}[theorem]{Lemma}
\newtheorem{remark}[theorem]{Remark}
\newtheorem{definition}[theorem]{Definition}

\newtheorem{proposition}[theorem]{Proposition}
\newenvironment{proof}[1][Proof]{\textbf{#1.} }{\ \rule{0.5em}{0.5em}}

\newcommand{\dru}{\delta_{ru}}
\newcommand{\dur}{\delta_{ur}}
\newcommand{\mm}{\mu_u}
\newcommand{\mr}{\mu_r}
\newcommand{\bu}{\beta_u}
\newcommand{\br}{\beta_r}
\newcommand{\gu}{\gamma_u}
\newcommand{\gr}{\gamma_r}
\newcommand{\pu}{\rho_u}
\newcommand{\pr}{\rho_r}

\begin{document}

\title{A two-patch epidemic model with nonlinear reinfection\\[2cm]
Un modelo epid\'emico de dos poblaciones con reinfecci\'on no lineal
}

\author{
Juan G. Calvo\thanks{CIMPA-Escuela de Matem\'atica, Universidad de Costa Rica, San Jos\'e, Costa Rica. 
E-Mail: \href{mailto: juan.calvo@ucr.ac.cr}{juan.calvo@ucr.ac.cr}}
\and % in case of need additional authors in the following format; en caso necesario los autores adicionales en el formato siguiente
{\sc Alberto Hern\'andez}\thanks{CIMPA-Escuela de Matem\'atica, Universidad de Costa Rica, San Jos\'e, Costa Rica.  
E-Mail: \href{mailto: albertojose.hernandez@ucr.ac.cr}{albertojose.hernandez@ucr.ac.cr}}
\and
{\sc Mason A. Porter}\thanks{Department of Mathematics, University of California Los Angeles, USA.  
E-Mail: \href{mailto: mason@math.ucla.edu}{mason@math.ucla.edu}}\and
{\sc Fabio Sanchez}\thanks{CIMPA-Escuela de Matem\'atica, Universidad de Costa Rica, San Jos\'e, Costa Rica.  
E-Mail: \href{mailto: fabio.sanchez@ucr.ac.cr}{fabio.sanchez@ucr.ac.cr}}
} % Name, Family Name, Address, City, Country. E-Mail.

\date{\small Received: xx-xx-xx; Revised: xx-xx-xx;\\ Accepted: xx-xx-xx}
 % Editor puts the dates; Las fechas las pone el editor

\maketitle
\newpage

\begin{abstract}

The propagation of infectious diseases and its impact on individuals play a major role in disease dynamics, and it is important to incorporate population heterogeneity into efforts to study diseases. As a simplistic but illustrative example, we examine interactions between urban and rural populations in the dynamics of disease spreading. Using a compartmental framework of susceptible--infected--susceptible ($\mathrm{SI\widetilde{S}}$) dynamics with some level of immunity, we formulate a model that allows nonlinear reinfection. We investigate the effects of population movement in the simplest scenario: a case with two patches, which allows us to model movement between urban and rural areas. To study the dynamics of the system, we compute a basic reproduction number for each population (urban and rural). We also compute steady states, determine the local stability of the disease-free steady state, and identify conditions for the existence of endemic steady states. From our analysis and computational experiments, we illustrate that population movement plays an important role in disease dynamics. In some cases, it can  be rather beneficial, as it can enlarge the region of stability of a disease-free steady state.

\end{abstract}

\KW Dynamical systems; population dynamics; mathematical modeling; biological contagions; population movement.

\begin{resumen}
La propagaci\'on de enfermedades infecciosas y su impacto en individuos juega un gran rol en la din\'amica de la enfermedad, y es importante incorporar heterogeneidad en la poblaci\'on en los esfuerzos por estudiar enfermedades. De manera simpl\'istica pero ilustrativa, se examina interacciones entre una poblaci\'on urbana y una rural en la din\'amica de la propagaci\'on de una enfermedad. Utilizando un sistema compartimental de din\'amicas entres susceptibles--infectados--susceptibles  ($\mathrm{SI\widetilde{S}}$) con cierto nivel de inmunidad, se formula un modelo que permite reinfecciones no lineales. Se investiga los efectos de movimiento de poblaciones en el escenario m\'as simple: un caso con dos poblaciones, que permite modelar movimiento entre un \'area urbana y otra rural. Con el fin de estudiar la din\'amica del sistema, se calcula el n\'umero b\'asico reproductivo para cada comunidad (rural y urbana). Se calculan tambi\'en puntos de equilibrio, la estabilidad local del estado libre de enfermedad, y se identifican condiciones para la existencia de estados de equilibrio end\'emicos. Del an\'alisis y experimentos computacionales, se ilustra que el movimiento en la poblaci\'on jueva un rol importante en la din\'amica del sistema. En algunos casos, puede ser beneficioso, pues incrementa la regi\'on de estabilidad del punto de equilibrio del estado libre de infecci\'on.
\end{resumen}

\PC Sistemas dinamicos, din\'amica de poblaciones; modelado matem\'atico; contagios biol\'ogicos; movimiento de poblaciones. 

\AMS 92D25, 92D30

%%%%%%%

%%%%%%%

\section{Introduction}

It is relatively easy for individuals to move between towns, cities, countries, and even continents; and incorporating movement between populations has become increasingly prevalent in the modeling and analysis of disease spreading \cite{ccc,pastor2015}. It is also important to consider movement in which individuals travel to a distinct location from their place of origin and then return to their original location in a relative short time. Such movement can lead to rapid spreading of infectious diseases, and examination of connected environments can give clues about the types of strategies that are needed for control of disease propagation \cite{bichara2016,martens2000}.

In 2003, the Severe Acute Respiratory Syndrome (SARS) epidemic was a major concern among public health officials worldwide \cite{cdc2018a}. This new infectious disease spread rapidly, and scientists and researchers scrambled to try to discern how to contain its spread (e.g., by reducing the spreading rate) and to seek treatments and a vaccine. The best control measure that was found at the time was to isolate individuals who had been in contact with infected individuals. The rapid spread of the disease was associated with movement of a doctor who was identified as ``Patient 0'' for SARS \cite{chowell2003,sanchez2005}. Population movement has also played an important role in subsequent events, such as the spread of ebola to the western hemisphere \cite{cdc2018b}, the spread of measles in some parts of the world by travelers \cite{cdc2018c}, and the resurgence of malaria through the massive migration of Nicaraguans to the northern part of Costa Rica \cite{malariacr}.

The use of compartmental models to describe the spreading of diseases has been explored thoroughly in numerous scenarios \cite{ccc,pastor2015}. For example, when there is nonlinear reinfection, an individual who was infected previously can become infected again through contact with an infectious individual after losing immunity \cite{treno2007}. Several models that allow individuals to lose immunity and become infectious again also exhibit a backward bifurcation in which a stable endemic equilibrium coexists with a stable disease-free equilibrium when the associated basic reproduction number is smaller than $1$  \cite{feng2000,sanchez2007,song2006,song2013,sanchez2018}. Moreover, in social contagion processes (e.g., spread of use of drugs, adoption of products, and so on), after an initial ``contagion'', some models include a backward bifurcation that can arise via social inputs \cite{song2006}. 

In the present paper, we generalize the compartmental model from~\cite{sanchez2007}, who studied a continuous dynamical system 
(in the form of coupled ordinary differential equations) that describes interactions between susceptible and infected individuals with the possibility of reinfection after loss of immunity, by incorporating population movement between urban and rural environments. Taking population movement into account is important for studies of disease dynamics in practice, and it changes the qualitative dynamics of disease spreading. We explore the simplest case, in which a population has two patches, and yield insights that will be useful for later explorations of disease spreading in a population that includes a network of patches.

Our paper proceeds as follows. In \Cref{sec:model}, we describe our compartment model of disease spreading (including nonlinear reinfection) between urban and rural environments. In \Cref{sec:analysis}, we give a formula for the model's basic reproduction number $\mathcal{R}_0$, analyze the existence and local stability of the disease-free state, and study the existence of endemic equilibria. In \Cref{sec:sims}, we illustrate several example scenarios with numerical computations. Finally, in \Cref{sec:disc}, we conclude and discuss the biological insights of our model.

%%%%

\section{A two-patch compartmental model} \label{sec:model}

We present a two-patch model of disease spreading in humans that incorporates nonlinear reinfection and population movement. Specifically, we generalize the model by Sanchez et al. \cite{sanchez2007} by incorporating the idea of urban versus rural environments. We use the subscript $u$ for urban variables and parameters and the subscript $r$ for rural variables and parameters. For $j\in\lbrace u,r\rbrace$, let $S_j$, $I_j$, and $\widetilde{S}_j$ denote the numbers of susceptible, infected, and post-recovery susceptible individuals, respectively. New susceptible individuals enter the system in proportion to the total population $N_j$, where $N_j=S_j+I_j+\widetilde{S}_j$. Let $\mu_j$ denote the rate of both births and deaths (for simplicity, we assume that they are the same.) Susceptible individuals become infected at a rate of $\beta_j$, infected individuals transition to a state of post-recovery susceptibility at a rate of $\gamma_j$, and post-recovery susceptible individuals can be reinfected at a rate of $\rho_j$. Such reinfection corresponds to infectious diseases (such as tuberculosis, malaria, and others \cite{glynn2008,grun1983}) in which subsequent infections are possible after loss of immunity. Typically, when there exists the possibility of reinfection, an initial infection tends to produce stronger symptoms \cite{crutcher1996}. We model movement between patches using the functions $\delta_{ij}(t)$, which denote the fraction of individuals who travel from patch $i \in\lbrace u,r\rbrace$ to patch $j \in\lbrace u,r\rbrace$ (with $i \neq j$) at time $t$. 

Our model consists of the following coupled system of ordinary differential equations: 
\begin{equation} \label{eq:sys}
	\begin{array}{rcl}
\dfrac{d S_u}{dt} &=& \mu_u N_u - \beta_u \dfrac{I_u}{N_u} S_u -\mu_u S_u + \dru S_r - \dur S_u\,,\\
\dfrac{d I_u}{dt} &=& \beta_u \dfrac{I_u}{N_u} S_u - (\mu_u+\gamma_u)I_u + \rho_u \dfrac{I_u}{N_u} \widetilde{S}_u +\dru I_r - \dur I_u\,,\\ 
\dfrac{d \widetilde{S}_u}{dt} &=& \gamma_u I_u - \rho_u \dfrac{I_u}{N_u} \widetilde{S}_u  - \mu_u \widetilde{S}_u +\dru \widetilde{S}_r - \dur \widetilde{S}_u\,,\\
\dfrac{d S_r}{dt} &=& \mu_r N_r - \beta_r \dfrac{I_r}{N_r} S_r -\mu_r S_r + \dur S_u - \dru S_r\,,\\
\dfrac{d I_r}{dt} &=& \beta_r \dfrac{I_r}{N_r} S_r - (\mu_r+\gamma_r)I_r + \rho_r \dfrac{I_r}{N_r} \widetilde{S}_r +\dur I_u - \dru I_r\,,\\ 
\dfrac{d \widetilde{S}_r}{dt} &=& \gamma_r I_r - \rho_r \dfrac{I_r}{N_r} \widetilde{S}_r  - \mu_r \widetilde{S}_r +\dur \widetilde{S}_u - \dru \widetilde{S}_r\,,\\
	\end{array}
\end{equation}
with initial conditions $S_j(0)=S_{j0}$, $I_j(0)=I_{j0}$, and $\widetilde{S}_j(0)=\widetilde{S}_{j0}$ (for $j\in\lbrace u,r\rbrace$). 

By adding the first three and last three equations in \eqref{eq:sys}, we see that $N_u$ and $N_r$ satisfy the linear dynamical system
\begin{equation*}
	\dfrac{d}{dt}\left[ 
\begin{array}{c}
N_u \\ 
N_r
\end{array} \right] = 
\left[
\begin{array}{rr}
-\dur & \dru \\ 
\dur & -\dru
\end{array}
\right]
\left[ 
\begin{array}{c}
N_u \\ 
N_r
	\end{array} \right]\,.
\end{equation*}
Its solution is
\begin{align*}
	N_u(t) &= e^{-\int_0^t \delta(s) ds}\left[ N_{u,0}\  + (N_{u,0}+N_{r,0})\ \int_0^t \dru(s) e^{\int_0^s \delta(h) dh}\ ds\right]\,,\\
N_r(t) &= e^{-\int_0^t \delta(s) ds}\left[ N_{r,0}\  + (N_{u,0}+N_{r,0})\ \int_0^t \dur(s) e^{\int_0^s \delta(h) dh}\ ds\right]\,,
\end{align*}
where $\delta(t) = \dur(t)+\dru(t)$ is the net movement of individuals at time $t$. The initial conditions are $N_{u,0}  = N_u(0)$ and $N_{r,0}  = N_r(0)$. 

In our analysis, we assume at first that $\dru$ and $\dur$ are constant, such that $N_u$ and $N_r$ simplify to 
\begin{align*}
	N_u(t) &= N_{u,0} e^{-t (\dur+\dru)} + (N_{u,0}+N_{r,0})\dfrac{\dru}{\dru+\dur} \left(1-e^{-t (\dur+\dru)}\right)\,, \\
N_r(t) &= N_{r,0} e^{-t (\dur+\dru)} + (N_{u,0}+N_{r,0})\dfrac{\dur}{\dru+\dur} \left(1-e^{-t (\dur+\dru)}\right)\,. 
\end{align*}
There are three cases:
\begin{enumerate}
    \item There is no movement (i.e., $\dur = \dru = 0$), which corresponds to the case of independent patches that was studied in \cite{sanchez2007}. 
     In this case, each patch has a backward bifurcation, and the steady state can depend on the number of initially infected individuals. 
    \item Movement occurs exclusively from one patch to the other. For example, suppose that $\dur = 0$ and $\dru>0$. In this case, $\displaystyle\lim_{t\to\infty} N_r(t) = 0$ (i.e., eventually, the entire population is in $u$). There is a backward bifurcation in this case as well, but now there is a total population of $N_{u,0}+N_{r,0}$ in the urban patch. We explore this case in Example \ref{ex:1} in Section \ref{sec:sims}.      
    \item There is movement in both directions between the two patches (i.e., $\dur>0$ and $\dru>0$). This is the primary scenario (and the principal novel contribution) of the present paper.
\end{enumerate}

%%%%

\section{Analysis of our model} \label{sec:analysis}

\subsection{Disease-free steady state and basic reproduction number} \label{sec:dfe}

Because the total population is constant, we can eliminate $\widetilde{S}_r$ from the dynamical system \eqref{eq:sys}. We then compute the Jacobian matrix of \eqref{eq:sys} and evaluate it at the disease-free steady state
\begin{equation}\label{free}
	(S_u^*,I_u^*,\widetilde{S}_u^*,S_r^*,I_r^*,\widetilde{S}_r^*)=(S_u^*,0,0,S_r^*,0,0) 
\end{equation}
to obtain the matrix
\begin{equation} \label{blah}
	\left[
\begin{array}{rrrrr}
-\dur & \mm-\bu &\mm & \dru & 0\\
0& \eta_{ur} & 0 & 0 & \dru \\
-\dru & \gu -\dru & -\dru-\dur-\mm &-\dru & -\dru\\
\dur-\mr& -\mr & -\mr &-\mr -\dru &-\br\\
0& \dur &0 &0 & \eta_{ru} \\
\end{array}
\right]\,,
\end{equation}
where $\eta_{ur} = \bu - \dur - \gu - \mm$ and $\eta_{ru} = \br - \dru - \gr - \mr$. For the dynamical system \eqref{eq:sys}, equilibria with $I_u=I_r=0$ necessarily also satisfy $\widetilde{S}_u=\widetilde{S}_r=0$.

The five eigenvalues of the matrix \eqref{blah} are
\begin{align*}
    \lambda_1 &= -(\dru+\dur)\,,\\
    \lambda_{2\pm} &= \frac{1}{2} \left[-(\dru + \dur + \mr + \mm) \pm \sqrt{(\dru - \dur + \mr - \mm)^2 + 4 \dru \dur} \right]\,,\\
    \lambda_{3\pm} &= \frac{1}{2}\left[ \eta_{ur}+\eta_{ru} \pm \sqrt{\left( \eta_{ur}-\eta_{ru} \right)^2 +4\dru\dur} \right]\,.\\
\end{align*}
All eigenvalues are real, and $\lambda_1$ and $\lambda_{2\pm}$ are negative. The two remaining eigenvalues $\lambda_{3\pm}$ are negative as long as
\begin{align} \label{eq:condEigNeg}
    \dru \dur &< \eta_{ur}\eta_{ru}\,, \quad  \eta_{ur}+ \eta_{ru} < 0 \,.
\end{align}

When there is only one population (i.e., one patch), the basic reproduction number is $\mathcal{R}_0=\frac{\beta}{\mu+\gamma}$. For our multiple-patch case, we define a ``local basic reproduction number'' for each patch:
\begin{equation*}
	\mathcal{R}_{0u} = \dfrac{\bu}{\gu+\mm}\,, \quad \mathcal{R}_{0r} = \dfrac{\br}{\gr+\mr}\,.
\end{equation*}	
We can then express the conditions in \eqref{eq:condEigNeg} as
\begin{subequations}\label{eq:condR0}
\begin{align} 
	\mathcal{R}_{0u} &< 1 + \dfrac{\dur}{\mm+\gu}\,, \label{eq:condR0a}\\
\mathcal{R}_{0r} &< 1 + \dfrac{\dru}{\mr+\gr},\label{eq:condR0b}\\ \dfrac{\dur}{\mm+\gu}\dfrac{\dru}{\mr+\gr} &< \left(\mathcal{R}_{0u}-1-\dfrac{\dur}{\mm+\gu}\right)\left(\mathcal{R}_{0r}-1-\dfrac{\dru}{\mr+\gr}\right)\,. \label{eq:condR0c}
\end{align}
\end{subequations}
See Figure \ref{fig:stab_reg} for an illustration of a typical region in which all eigenvalues are negative. For progressively smaller $\dur$ and $\dru$, the shaded region approaches the unit square.

\begin{figure}[htb!]
\centering
\includegraphics[width=0.7\linewidth]{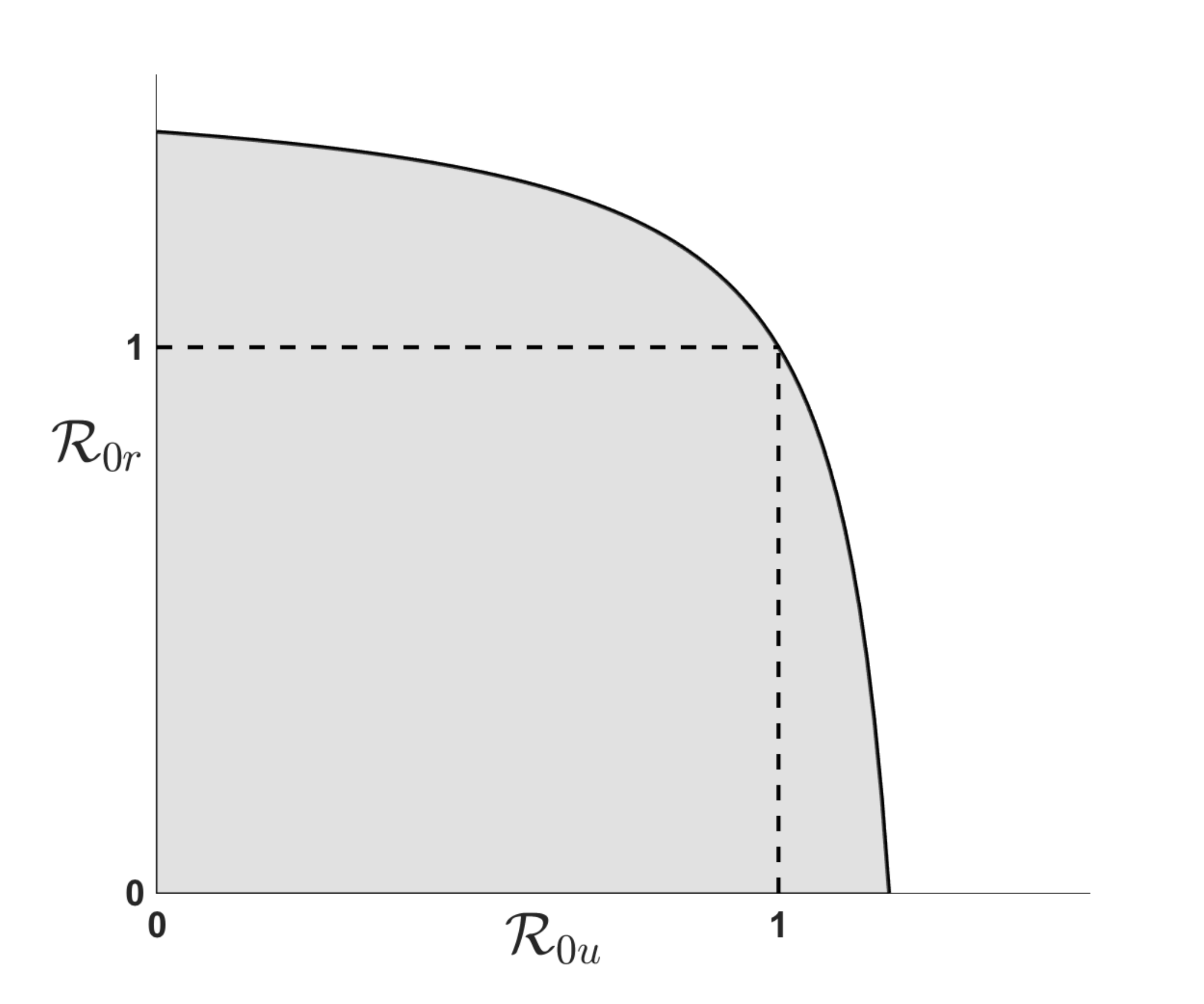}
\caption{A typical region of local asymptotic stability of the disease-free steady state \eqref{free} of the dynamical system \eqref{eq:sys} from the conditions \eqref{eq:condR0} as a function of the local basic reproduction numbers $\mathcal{R}_{0u}$ and $\mathcal{R}_{0r}$. \label{fig:stab_reg}}
\end{figure}

We have thus established the following lemma.

\begin{lemma} \label{lem:diseaseFree}
Assume that Eqs.~\eqref{eq:condR0} hold. It then follows that the disease-free steady state \eqref{free} of the dynamical system \eqref{eq:sys} is locally asymptotically stable.
\end{lemma}

\begin{remark} {\rm
The conditions in \eqref{eq:condR0} are satisfied when $\mathcal{R}_{0u}<1$ and $\mathcal{R}_{0r}<1$. We then have local stability in the rural and urban patches if we treat them as independent. Additionally, from Lemma \ref{lem:diseaseFree}, we see that it is possible to obtain local asymptotic stability for the disease-free steady state \eqref{free} even when one or both local basic reproduction numbers are larger than $1$. In such a scenario, movement is beneficial, as it leads to local asymptotic stability of the disease-free steady state in situations that would not be the case for independent patches. We illustrate such a scenario in Example \ref{ex:2} in Section \ref{sec:sims}. 
}
\end{remark}

%%%

\subsection{Technical tool: The Poincar\'e--Miranda theorem}
\label{sec:ee}

As preparation for analyzing the existence of steady states in our model, we briefly recall some classical results by Poincar\'e and others. See \cite{frankowska2018, kulpa1997, szyman2015, turzanski2012} for detailed accounts of the relevant theory. In 1817, Bolzano proved the following well-known theorem:
\begin{theorem}
Let $f:[a,b]\longrightarrow \mathbb{R}$ be a continuous function such that $f(a) \cdot f(b)< 0$. There then exists $c \in (a,b)$ such that $f(c) = 0$.
\end{theorem}

\begin{definition}
Let $I^n = [0,1]^n$, and let $\partial I^n$ denote its boundary. For each $i \in \{1, \dots, n\}$, let 
\begin{align*}
    I_i^- = \{x \in I^n | x_i = 0\} \,, \qquad
    I_i^+ = \{x\in I^n | x_i = 1 \}\,,
\end{align*}
be the opposite $i$-th faces of the boundary $\partial I^n$.
\end{definition}

In 1883--1884, Poincar\'e announced a generalization of Bolzano's theorem without providing a proof. 

\begin{theorem}[Poincar\'e] 
Let $F: I^n \longrightarrow \mathbb{R}^n$, with $F = (F_1,\cdots, F_n)$, be a continuous map such that
\begin{equation*} 
	F_i(I_{i}^-) \subseteq (-\infty, 0]
\end{equation*}
and 
\begin{equation*}
	F_i(I_{i}^+) \subseteq [0, +\infty)
\end{equation*}
for every $i \in \{1, \dots, n\}$. It then follows that there exists $c\in I^n$ such that $F(c) = 0$.
\end{theorem}

In the 1940s, Miranda rediscovered Poincar\'e's theorem and showed that it is logically equivalent to Brower's fixed-point theorem. Since then, this result has often been called the Poincar\'e--Miranda theorem. We require a modified version of the Poincar\'e--Miranda theorem in two dimensions.

\begin{proposition} \label{prop:miranda}
Let $F: [0,1]^2 \longrightarrow \mathbb{R}^2$ such that $F(x,y)= (F_1(x,y),F_2(x,y))^T$ is continuous and
$F(0,0) = 0$. Assume that
\begin{equation*}
\begin{array}{cc}
	F_1(x,0) > 0 \quad \mathrm{for\,\, all}\,\,\, x\in (0,1]\,,	& F_1(x,1) <0\ \quad \mathrm{for\,\, all}\,\,\, x\in [0,1]\,,\\
	F_2(0,y) > 0 \quad \mathrm{for\,\, all}\,\,\, y\in (0,1]\,,	& F_2(1,y) <0\ \quad \mathrm{for\,\, all}\,\,\, y\in [0,1]\,.
\end{array}
\end{equation*}
Assume additionally that $ \partial{F_1} / \partial{y}$ and $\partial{F_2} / \partial{x}$ are both continuous from the right at $(0,0)$, with $\partial{F_1} / \partial{y}(0,0)>0$ and $\partial{F_2} / \partial{x}(0,0)>0$. It then follows that there exists $(x_0,y_0) \in (0,1)^2$ such that $F(x_0,y_0) = (0,0)^T$.
\end{proposition}

 \begin{proof}
Because $F$ is continuous, both $F_1$ and $F_2$ are also continuous. Using the facts that $F_1(x,0) >0$ for all $x \in (0,1]$ and that $\partial {F_1} / \partial {x}$ is continuous from the right, it follows that there exists $\epsilon_0 > 0$ such that $\partial {F_1} /\partial x (x,0) > 0$ for all $x \in (0, \epsilon_0)$. By the implicit function theorem, for $x \in (0,\epsilon_0)$, one can write $y = g_1(x)$, where the function $g_1$ is differentiable. Therefore,
\begin{equation}
\label{inverse thm}
 	g_1'(0) = -\frac{ \frac{\partial F_1}{\partial x}(0,0) }{\frac{\partial F_1}{\partial y}(0,0)} \,. 
 \end{equation}
By the continuity of $F_1$ in the compact set $[0,1]^2$ and using Equation \eqref{inverse thm}, it follows that there exists $\epsilon_1 > 0$ such that $F_1(x,\epsilon_1) > 0$ for all $x \in [0,1]$.
By the same argument, there also exists $\epsilon_2 \in (0,1)$ such that $F_2(\epsilon_2,y)> 0$ for all $y \in [0,1]$. By the Poincar\'e--Miranda theorem, there must exist $(x_0,y_0) \in [\epsilon_1, 1]\times [\epsilon_2,1]$ such that $F(x_0,y_0) = (0,0)$.
 \end{proof}

%%%%

\subsection{Existence of multiple-population endemic steady states} 

We now examine endemic steady states, in which there are infected individuals at steady state in both the urban and rural patches.

We start by proving the following lemma.

\begin{lemma}
Suppose that $\mathcal{R}_{0u} > 1+\dfrac{\dur}{\mm+\gu}$ and $\mathcal{R}_{0r} > 1+\dfrac{\dru}{\mr+\gr}$. It then follows that exists at least one endemic state, for which both $I_u^*>0$ and $I_r^*>0$. This is, there are infected individuals at steady state in both the urban and the rural patches.
\end{lemma}

\begin{proof}
We deduce the existence of a solution of the following nonlinear system of algebraic equations:
\begin{equation} \label{eq:systemSS}
	\begin{array}{rcl}
0 &=& \mu_u N_u - \beta_u \dfrac{I_u}{N_u} S_u -\mu_u S_u + \dru S_r - \dur S_u\,,\\
0 &=& \beta_u \dfrac{I_u}{N_u} S_u - (\mu_u+\gamma_u)I_u + \rho_u \dfrac{I_u}{N_u} \widetilde{S}_u +\dru I_r - \dur I_u\,,\\ 
0 &=& \gamma_u I_u - \rho_u \dfrac{I_u}{N_u} \widetilde{S}_u  - \mu_u \widetilde{S}_u +\dru \widetilde{S}_r - \dur \widetilde{S}_u\,,\\
0 &=& \mu_r N_r - \beta_r \dfrac{I_r}{N_r} S_r -\mu_r S_r + \dur S_u - \dru S_r\,,\\
0 &=& \beta_r \dfrac{I_r}{N_r} S_r - (\mu_r+\gamma_r)I_r + \rho_r \dfrac{I_r}{N_r} \widetilde{S}_r +\dur I_u - \dru I_r\,,\\ 
0 &=& \gamma_r I_r - \rho_r \dfrac{I_r}{N_r} \widetilde{S}_r  - \mu_r \widetilde{S}_r +\dur \widetilde{S}_u - \dru \widetilde{S}_r\,.\\
	\end{array}
\end{equation}

Given two parameters $(B_u,B_r) \in [0,1]^2$, we consider the auxiliary linear system 
\begin{equation} \label{eq:systemSS2}
	\begin{array}{rcl}
0 &=& \mu_u N_u - \beta_u B_u S_u -\mu_u S_u + \dru S_r - \dur S_u\,,\\
0 &=& \beta_u B_u S_u - (\mu_u+\gamma_u)I_u + \rho_u B_u \widetilde{S}_u +\dru I_r - \dur I_u\,,\\ 
0 &=& \gamma_u I_u - \rho_u B_u \widetilde{S}_u  - \mu_u \widetilde{S}_u +\dru \widetilde{S}_r - \dur \widetilde{S}_u\,,\\
0 &=& \mu_r N_r - \beta_r B_r S_r -\mu_r S_r + \dur S_u - \dru S_r\,,\\
0 &=& \beta_r B_r S_r - (\mu_r+\gamma_r)I_r + \rho_r B_r \widetilde{S}_r +\dur I_u - \dru I_r\,,\\ 
0 &=& \gamma_r I_r - \rho_r B_r \widetilde{S}_r  - \mu_r \widetilde{S}_r +\dur \widetilde{S}_u - \dru \widetilde{S}_r\,.\\
	\end{array}
\end{equation}

Suppose that $(S_u,I_u,\widetilde{S}_u,S_r,I_r,\widetilde{S}_r)$ is a solution of the linear system \eqref{eq:systemSS2}. This solution also satisfies the nonlinear system \eqref{eq:systemSS} if $I_u/N_u = B_u$ and $I_r/N_r=B_r$. By adding all of the equations in \eqref{eq:systemSS2}, we see 
that $\dur N_u = \dru N_r$. We then rescale variables in \eqref{eq:systemSS2} by substituting 
\begin{align*}
	s_u &=\dfrac{S_u}{N_u}\,,\quad s_r =\dfrac{S_r}{N_r}\,,\\
i_u &=\dfrac{I_u}{N_u}\,, \quad  i_r =\dfrac{I_r}{N_r}\,,\\
	\widetilde{s}_u &=\dfrac{\widetilde{S}_u}{N_u}\,,\quad  \widetilde{s}_r =\dfrac{\widetilde{S}_r}{N_r}
\end{align*}
to obtain a linear system of algebraic equations with parameters $B_u$ and $B_r$. This system is
\begin{subequations} \label{eq:linearSystem}
\begin{align}
	\mu_u -\left(\beta_u B_u +\mu_u+\dur\right ) s_u + \dur s_r  &=  0\,,  \label{eq:su}\\
\beta_u B_u s_u - (\mu_u+\gamma_u + \dur ) i_u +  \rho_u  B_u \widetilde{s}_u +  \dur i_r &=  0\,,\label{eq:iu}\\ 
\gamma_u i_u - \rho_u  B_u \widetilde{s}_u - \left(\mu_u  + \dur  \right)\widetilde{s}_u + \dur \widetilde{s}_r &= 0\,,\label{eq:ru}\\
\mu_r - \left(\beta_r B_r +\mu_r+\dru\right ) s_r + \dru s_u  &= 0\,,\label{eq:sr}\\
\beta_r  B_r s_r - (\mu_r+\gamma_r + \dru) i_r  +\rho_r  B_r \widetilde{s}_r + \dru i_u &= 0\,,\label{eq:ir}\\ 
\gamma_r i_r - \rho_r  B_r \widetilde{s}_r - \left( \mu_r  + \dru  \right)\widetilde{s}_r + \dru \widetilde{s}_u &= 0\label{eq:rr}\,.
\end{align} 
\end{subequations}
We now solve Eqs.~\eqref{eq:su} and \eqref{eq:sr} to obtain
\begin{align*}
	s_u(B_u,B_r) &= \dfrac{\beta_r  \mu_u B_r+\dur \mu_r+\dru \mu_u+\mu_r\mu_u}{\left(\beta_u B_u +\mu_u+\dur\right )\left(\beta_r B_r +\mu_r+\dru\right )-\dru\dur}\,,\\
s_r(B_u,B_r) &= \dfrac{\beta_u  \mu_r B_u+\dur \mu_r+\dru \mu_u+\mu_r\mu_u}{\left(\beta_u B_u +\mu_u+\dur\right )\left(\beta_r B_r +\mu_r+\dru\right )-\dru\dur}\,.
\end{align*}
Using Eqs.~\eqref{eq:iu}, \eqref{eq:ru}, \eqref{eq:ir}, \eqref{eq:rr}, we solve for $i_u$ and $i_r$ to obtain
\begin{align*}
	i_u(B_u,B_r) &= \dfrac{1}{\Delta}\left( s_r \br \dur B_r  \Delta_{i_u}^{(1)} + s_u \bu B_u \Delta_{i_u}^{(2)} \right)\,,\\
i_r(B_u,B_r) &= \dfrac{1}{\Delta}\left( s_u \bu \dru B_u \Delta_{i_r}^{(1)} + s_r \br B_r \Delta_{i_r}^{(2)} \right)\,,
\end{align*}
where
\begin{align*}
	\Delta_{i_u}^{(1)} &= \dur \mr +  \dru \mm + \mm \mr  + \pr  (\mm + \dur) B_r +  \pu (\gr+ \mr+ \dru) B_u  +  \pu \pr B_u B_r\,,\\
\Delta_{i_u}^{(2)} &=  (\dru + \gr + \mr) (\dur \mr + \dru \mm + \mr \mu) + (\dru + \mr) (\dur + \mm) \pr B_r + \\
 & \hspace{1cm} + (\dru + \mr) (\dru + \gr + \mr) \pu B_u + (\dru + \mr) \pu \pr B_u B_r\,,\\
\Delta_{i_r}^{(1)} &= \dur \mr +  \dru \mm + \mm \mr  + \pu  (\mr + \dru) B_u +  \pr (\gu+ \mm+ \dur) B_r  +  \pu \pr B_u B_r\,,\\
\Delta_{i_r}^{(2)} &=  (\dur + \gu + \mu) (\dur \mr + \dru \mm + \mr \mu) + (\dur + \mm) (\dru + \mr) \pu B_u + \\
 & \hspace{1cm} + (\dur + \mm) (\dur + \gu + \mm) \pr B_r + (\dur + \mm) \pu \pr B_u B_r\,,\\
\Delta &= (\dur \mr + \dru \mm+\mr\mm) ( (\dru + \gr + \mr + B_r \pr) (\dur + \gu + \mm + B_u \pu) -\dur\dru)\,.
\end{align*}
One can write similar expressions for $\widetilde{s}_u(B_u,B_r)$ and $\widetilde{s}_r(B_u,B_r)$. The solution $(s_u,i_u,\widetilde{s}_u,s_r,i_r,\widetilde{s}_r)$ of \eqref{eq:linearSystem} satisfies the nonlinear algebraic system \eqref{eq:systemSS} if and only if $i_u(B_u,B_r)=B_u$ and $i_r(B_u,B_r)=B_r$. We then define
\begin{align*}
	G_1(B_u,B_r) &= i_u(B_u,B_r)-B_u\,,\\
G_2(B_u,B_r) &= i_r(B_u,B_r)-B_r\,.
\end{align*}
We will show that the system 
\begin{equation}\label{eq_systemG1G2}
G_1(B_u,B_r) = G_2(B_u,B_r)=0
\end{equation}
has at least one solution $(B_u,B_r)\in (0,1)^2$. This implies the existence of a solution of the nonlinear algebraic system \eqref{eq:systemSS} with $i_u>0$ and $i_r>0$. This solution corresponds to a steady state with a nonzero infected population in both the urban and the rural patches.

Consider the surface $z_1 = G_1(B_u,B_r)$. We have that $G_1(0,B_r) = i_u(0,B_r) \geq 0$ for $B_r\in [0,1]$, with equality when $B_r=0$. Because $G_1(0,0)=0$, we need to cut the origin to guarantee the existence of a positive solution of Eq.~\eqref{eq_systemG1G2}. A straightforward calculation yields
\begin{align*}
    \dfrac{\partial G_1}{\partial B_u}(0,0) &= \dfrac{-\dur (\gr + \mr) + (\bu -\gu - \mm) (\dru + \gr + \mr) }{\dur (\gr + \mr) + (\dru + \gr + \mr) (\gu + \mm)}\\
    &> \dfrac{\dur \dru }{\dur (\gr + \mr) + (\dru + \gr + \mr) (\gu + \mm)} >0\,,
\end{align*}
because $\bu>\gu+\mm+\dur$ by hypothesis. Analogously, we compute that 
\begin{equation}
	\dfrac{\partial G_2}{\partial B_r}(0,0)>0\,.
\end{equation}

We also compute that $G_1(1,B_r) = i_u(1,B_r) - 1<0$ for $B_r\in [0,1]$ and that $G_2(B_u,1)<0$ for $B_u\in [0,1]$. By Proposition \ref{prop:miranda}, we conclude that there exists $(B_u,B_r) \in (\epsilon_1,1)\times (\epsilon_2,1)$ such that $G_1(B_u,B_r) = G_2(B_u,B_r)=0$. Therefore, the nonlinear algebraic system \eqref{eq:systemSS} has a solution that corresponds to a steady state with positive values for both $I_u^*$ and $I_r^*$.
\end{proof}

\begin{figure}[htb!]
\subfloat[The curve $G_1(B_u,B_r)=0$ has a tangent with slope $M_u < 0$ at the origin.  \label{fig:lemmaEnd}]{\includegraphics[width=0.48\linewidth]{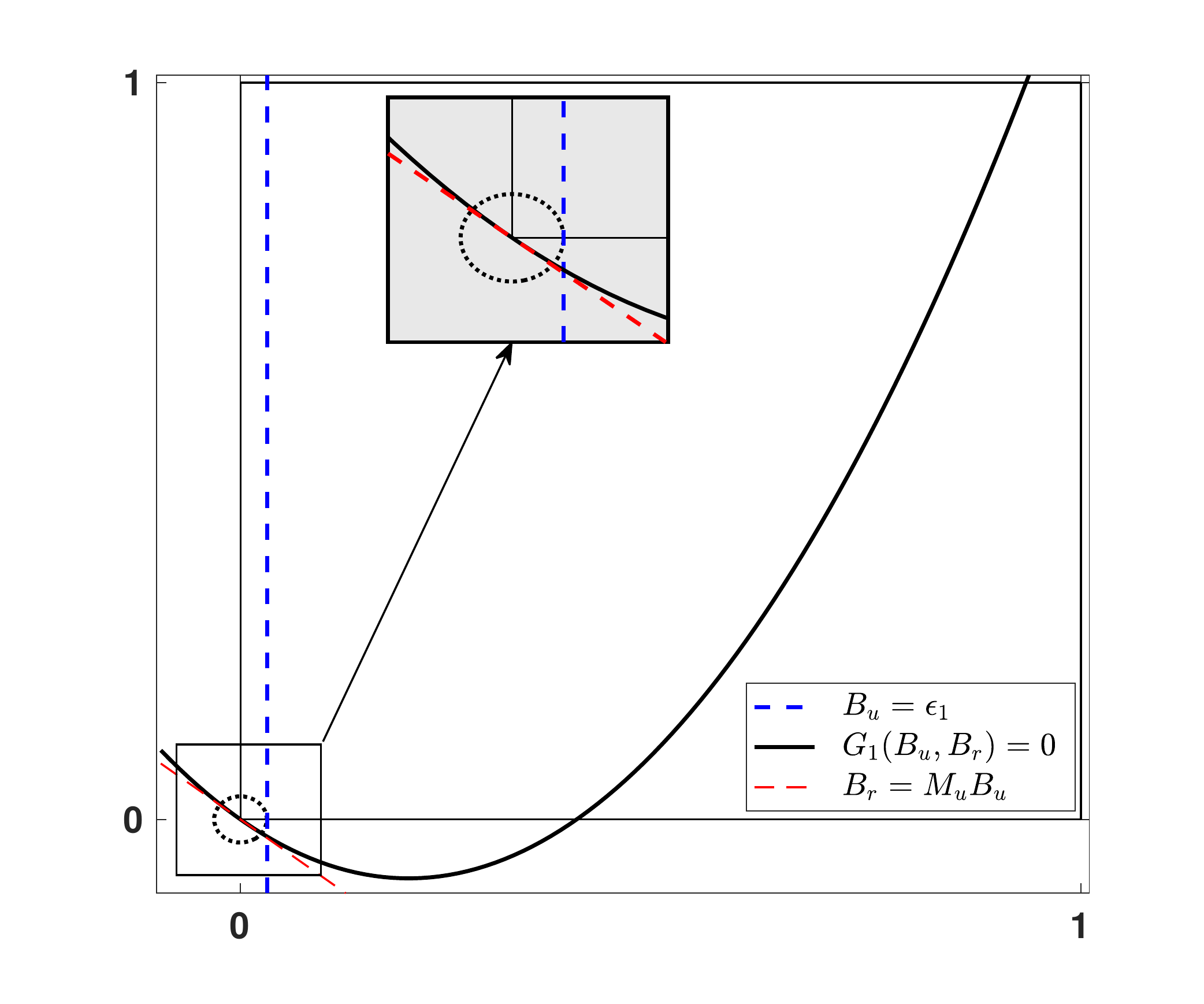}}
\hfill
\subfloat[The curves $G_1(B_u,B_r)=0$ and $G_2(B_u,B_r)=0$ intersect at two points with positive coordinates.
\label{fig:remarkEnd}]{\includegraphics[width=0.48\linewidth]{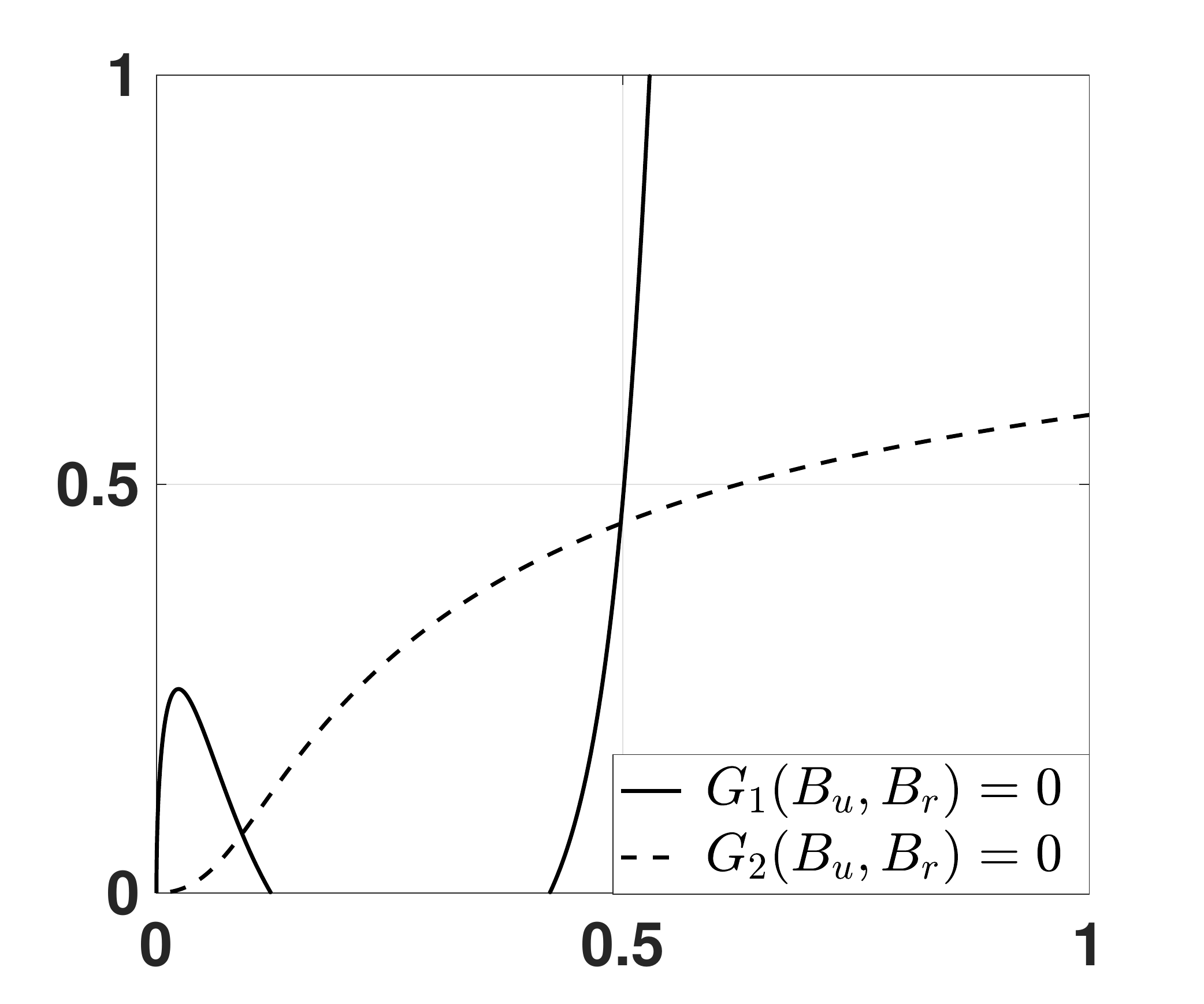}}
\caption{Endemic steady states in urban and rural environments. In these steady states, an infected population persists in both the rural and the urban patches.
}
\end{figure}

\begin{remark} {\rm \label{rem:severalPoints}
Using numerical computations, we can observe the existence of two or more distinct endemic steady states. However, we need to explore them further to characterize them; see Figure \ref{fig:remarkEnd} and Examples \ref{ex:3} and \ref{eq:ninePts} in Section \ref{sec:sims}.
}
\end{remark}

%%%%%

\section{Examples} \label{sec:sims}

We now present some numerical simulations of the dynamical system \eqref{eq:sys} for a variety of parameter values. Using these examples, we illustrate that population movement strongly influences how a disease can propagate.

\begin{example} {\rm \label{ex:1} %file ex1_dur0.m

We first explore the behavior of our model \eqref{eq:sys} when $\dur=0$ and $\dru\neq 0$, which describes movement in one direction (specifically, from the rural patch to the urban one). In some countries, it is common for individuals in rural areas to move to urban areas for work \cite{coffee2015}. This is a type of short-term mobility. We consider the following parameter values:
\begin{equation*}
	\begin{array}{cccc}
\mm = 1/(365\cdot 80)\,, & \pu = 0.08\,, & \gu = 0.01\,, & \bu = 0.03\,,\\
\mr = 1/(365\cdot 70)\,, & \pr = 0.04\,, & \gr = 0.01\,, & \br = 0.02\,,
\end{array}
\end{equation*}
with initial conditions 
\begin{equation*}
	\begin{array}{ccc}
S_{u0} = 999\,, & I_{u0} = 1\,, & \widetilde{S}_{u0} = 0\,,\\
S_{r0} = 300\,, & I_{r0} = 0\,, & \widetilde{S}_{r0} = 0\,.
\end{array}
\end{equation*}
In this case,
\begin{align*}
	\mathcal{R}_{0u} > 1 \quad \text{and} \quad \mathcal{R}_{0r} \approx 2 \,.
\end{align*}
In Figure \ref{fig:ex1_dur_change}, we show $I_u(t)$ as we vary $\dru$. We show the solutions for both the urban and the rural patches when $\dru = 0$ (i.e., when there is no movement) and $\dru = 0.01$ in Figure \ref{fig:ex1}. We observe that movement in one direction increases the number of infected individuals and that larger values of $\dru$ affects only the approach speed to the steady state.

\begin{figure}[htb!]
\centering%no!\hfill
\includegraphics[width=.8\linewidth]{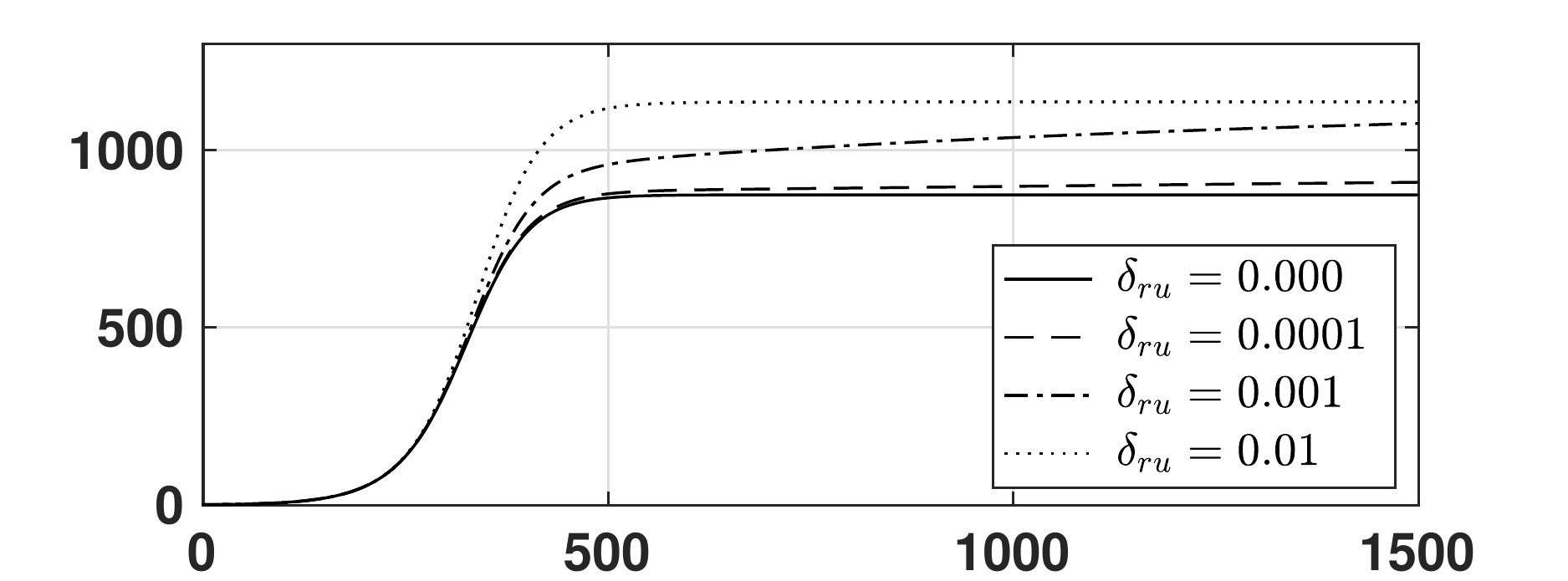}
\caption{The number $I_u$ of infected individuals in the urban patch as a function of time (in days) for $\dur = 0$ and different values of $\dru$; see Example \ref{ex:1}. 
 \label{fig:ex1_dur_change}}
\end{figure}

\begin{figure}[htb!]
\centering
\subfloat[Urban patch with $\dru=0$ \label{fig:ex1_u1}]{\includegraphics[width=0.5\linewidth]{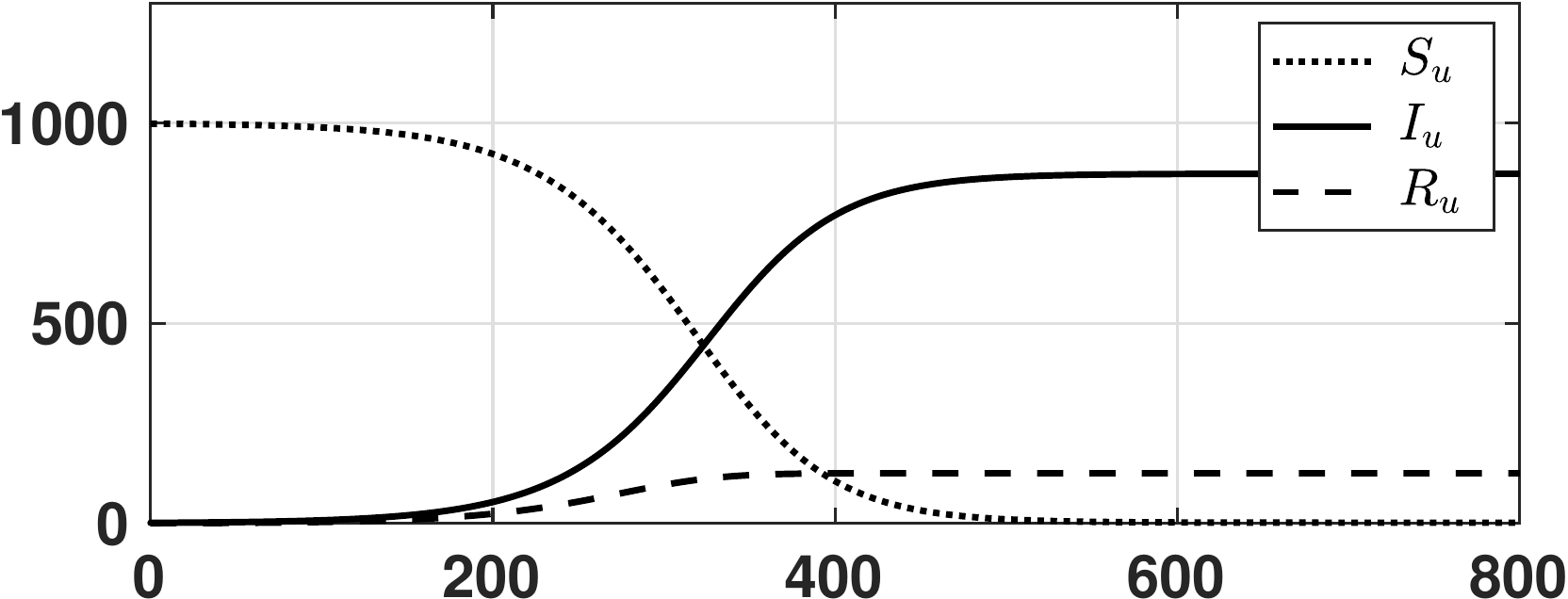}}
\hfill
\subfloat[Rural patch with $\dru=0$ \label{fig:ex1_r1}]{\includegraphics[width=0.5\linewidth]{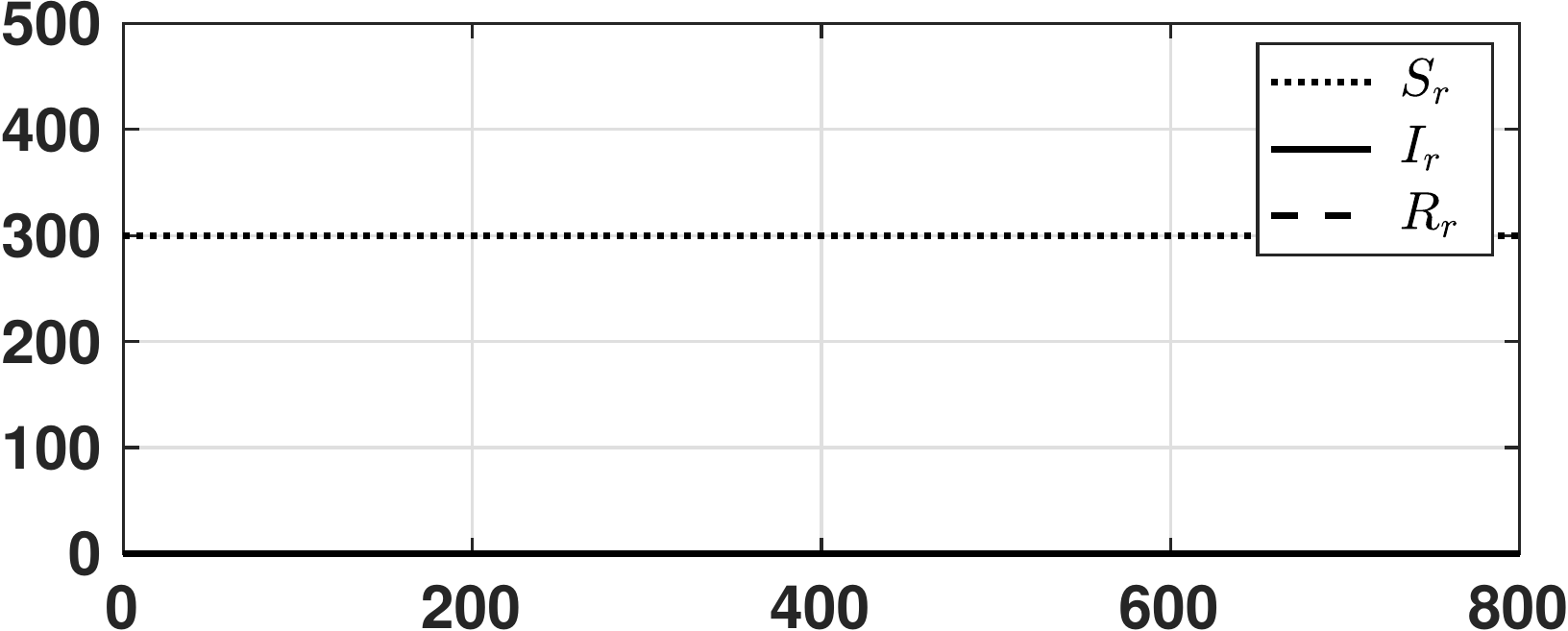}}
\hfill
\subfloat[Urban patch with $\dru=0.01$ \label{fig:ex1_u2}]{\includegraphics[width=0.5\linewidth]{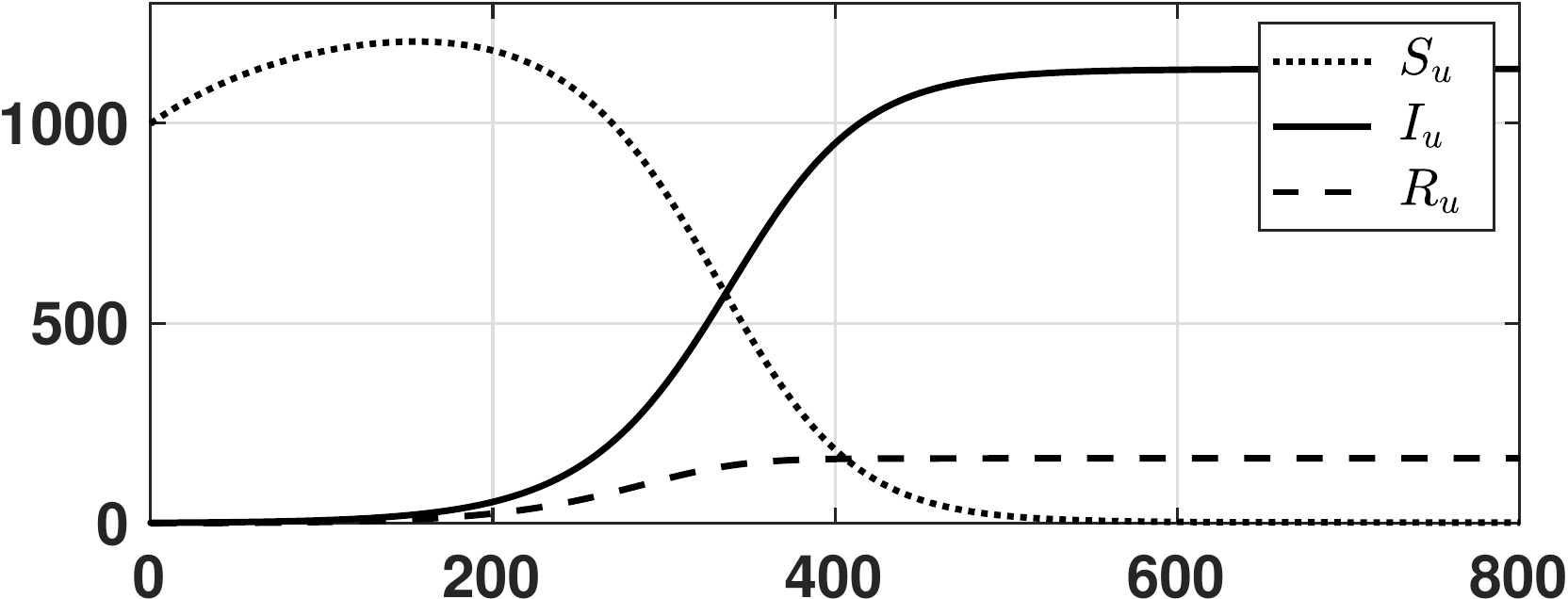}}
\hfill
\subfloat[Rural patch with $\dru=0.01$ \label{fig:ex1_r2}]{\includegraphics[width=0.5\linewidth]{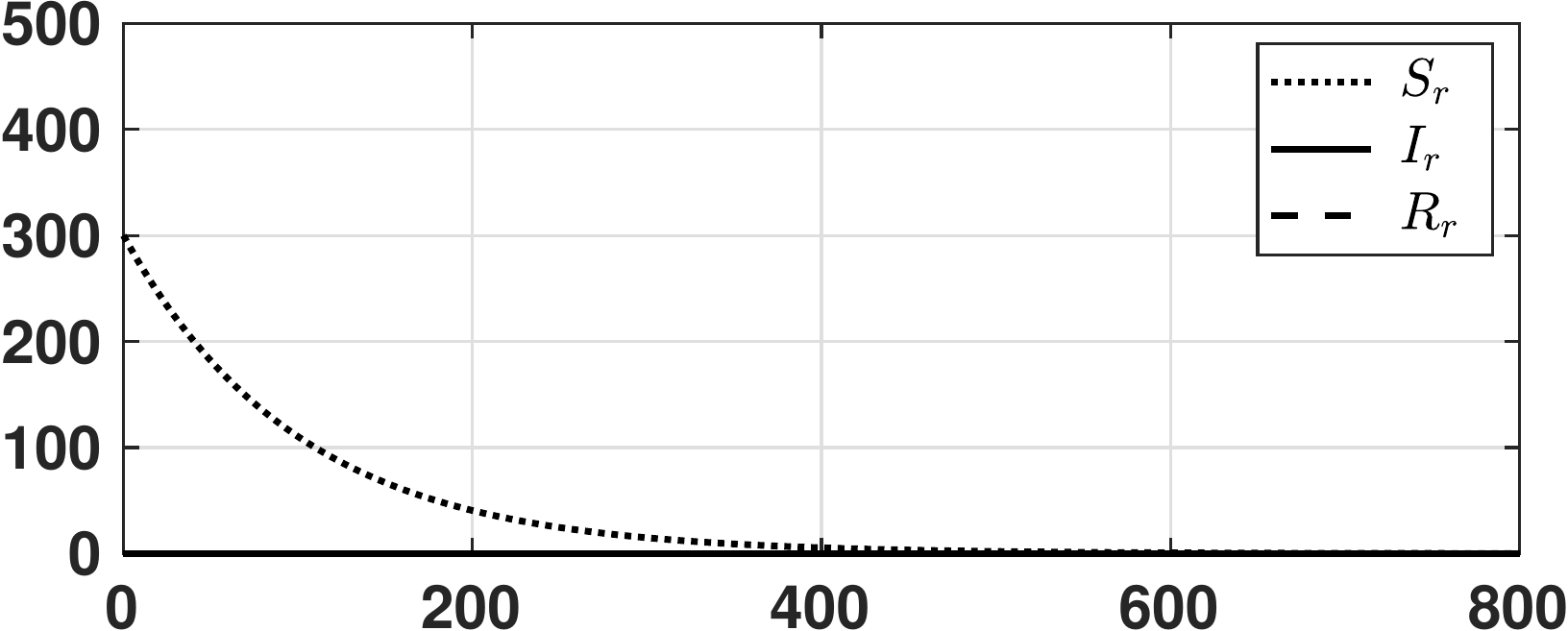}}
\caption{Effect of movement in just one direction (with $\dur=0$) for different rates of immigration $\dru$; see Example \ref{ex:1}. 
\label{fig:ex1}}
\end{figure}

}
\end{example}

\begin{example} \label{ex:2} % file ex4.m
{\rm
We now compare the effects of conditions \eqref{eq:condR0a} and \eqref{eq:condR0b} to the standard condition $\mathcal{R}_0<1$ for local asymptotic stability of the disease-free steady state \eqref{free}. We use the parameter values
\begin{equation*}
	\begin{array}{cccc}
\mm = 1/(365\cdot 80)\,, & \pu = 0.08\,, & \gu = 0.01\,, & \bu = 3\cdot 10^{-2}\,,\\
\mr = 1/(365\cdot 70)\,, & \pr = 0.40\,, & \gr = 0.10\,, & \br = 2\cdot 10^{-5}
	\end{array}
\end{equation*}
and initial conditions 
\begin{equation*}
	\begin{array}{ccc}
S_{u0} = 999, & I_{u0} = 1\,, & \widetilde{S}_{u0} = 0\,,\\
S_{r0} = 300, & I_{r0} = 0\,, & \widetilde{S}_{r0} = 0\,.
	\end{array}
\end{equation*}
We calculate that $\mathcal{R}_{0u} \approx 2.9$ and $\mathcal{R}_{0r} \approx 2\cdot 10^{-4}$. Therefore, in the absence of movement, disease persists in the urban patch but dies out in the rural patch; see Figures \ref{fig:ex2_u1} and \ref{fig:ex2_r1}. For $\dur=\dru=0.05$, we calculate that $\dur/(\gu+\mu) \approx 5.98$, so conditions \eqref{eq:condR0} are satisfied. Therefore, the disease-free steady state is locally asymptotically stable; see Figures \ref{fig:ex2_u2} and \ref{fig:ex2_r2}.

In this example, there exists one endemic state, with $(I_u^*/N_u^*, I_r^*/N_r^*) \approx (0.82, 0.76)$. We study two variations:
\begin{enumerate}
    \item[(1)]When the initial number of infected individuals increases, we reach the endemic steady state. See Figure \ref{fig:ex2c1} for an illustration with initial conditions $S_{u0} = 900$ and $I_{u0} = 100$. 
    \item[(2)]When we slightly increase $\bu$ to $\bu = 5.3\cdot 10^{-2}$, condition \eqref{eq:condR0c} is no longer satisfied, and the system has an endemic state with just one initial infected individual in the urban patch; see Figure \ref{fig:ex2c2}. 
\end{enumerate}

\begin{figure}[htb!]
\subfloat[Urban patch \label{fig:ex2_u1}]{\includegraphics[width=0.5\linewidth]{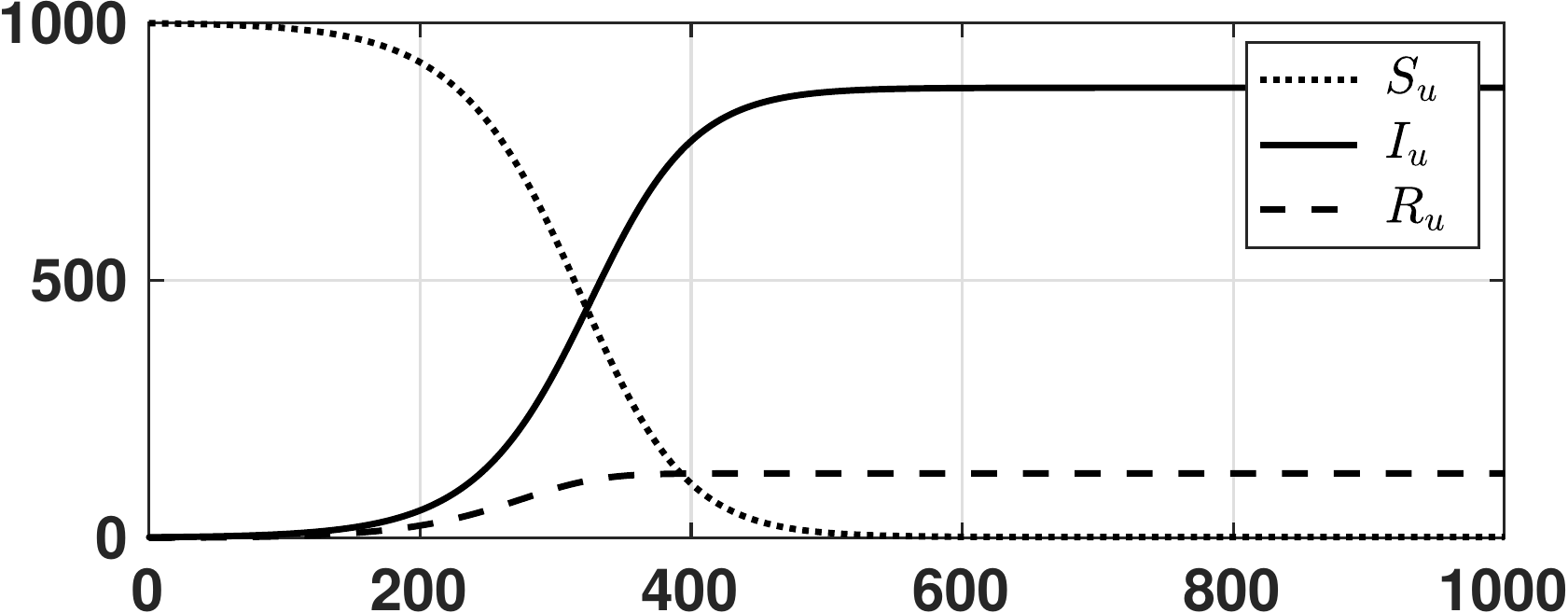}}
\hfill
\subfloat[Rural patch \label{fig:ex2_r1}] {\includegraphics[width=0.5\linewidth]{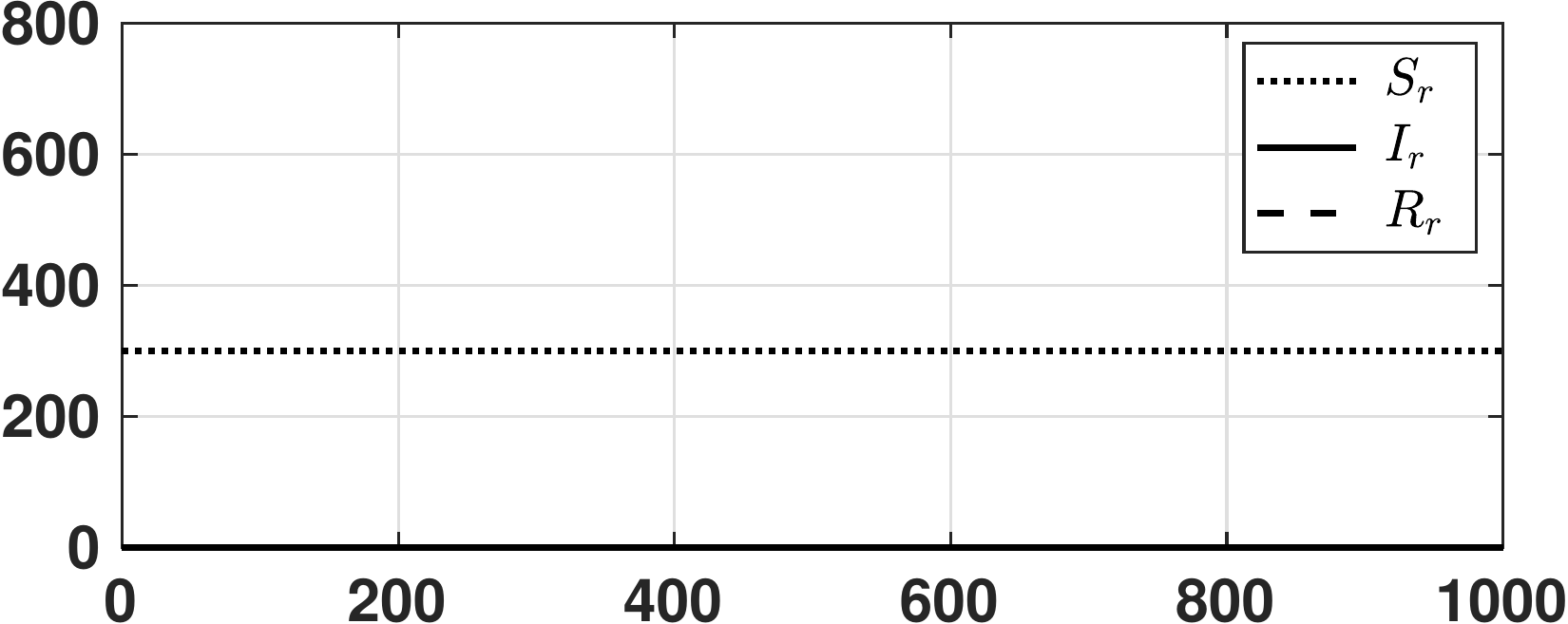}}
\hfill
\subfloat[Urban patch \label{fig:ex2_u2}]{\includegraphics[width=0.5\linewidth]{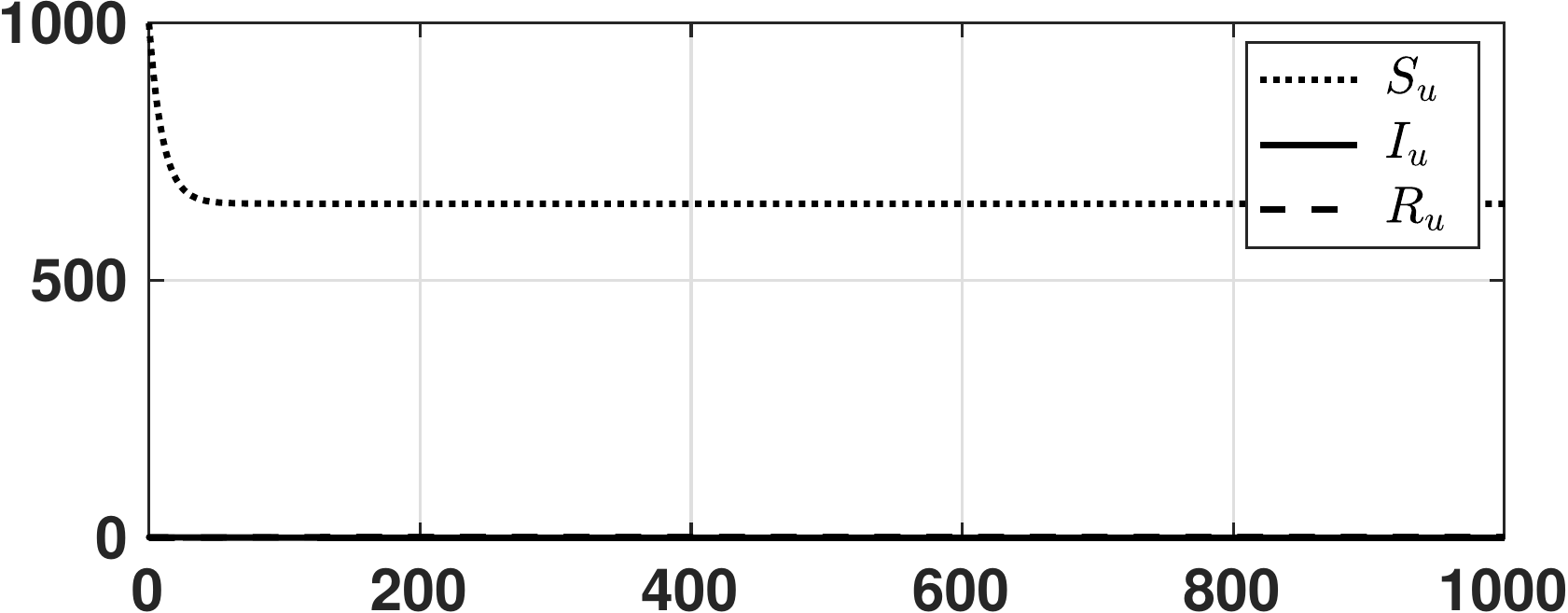}}
\hfill
\subfloat[Rural patch \label{fig:ex2_r2}]{\includegraphics[width=0.5\linewidth]{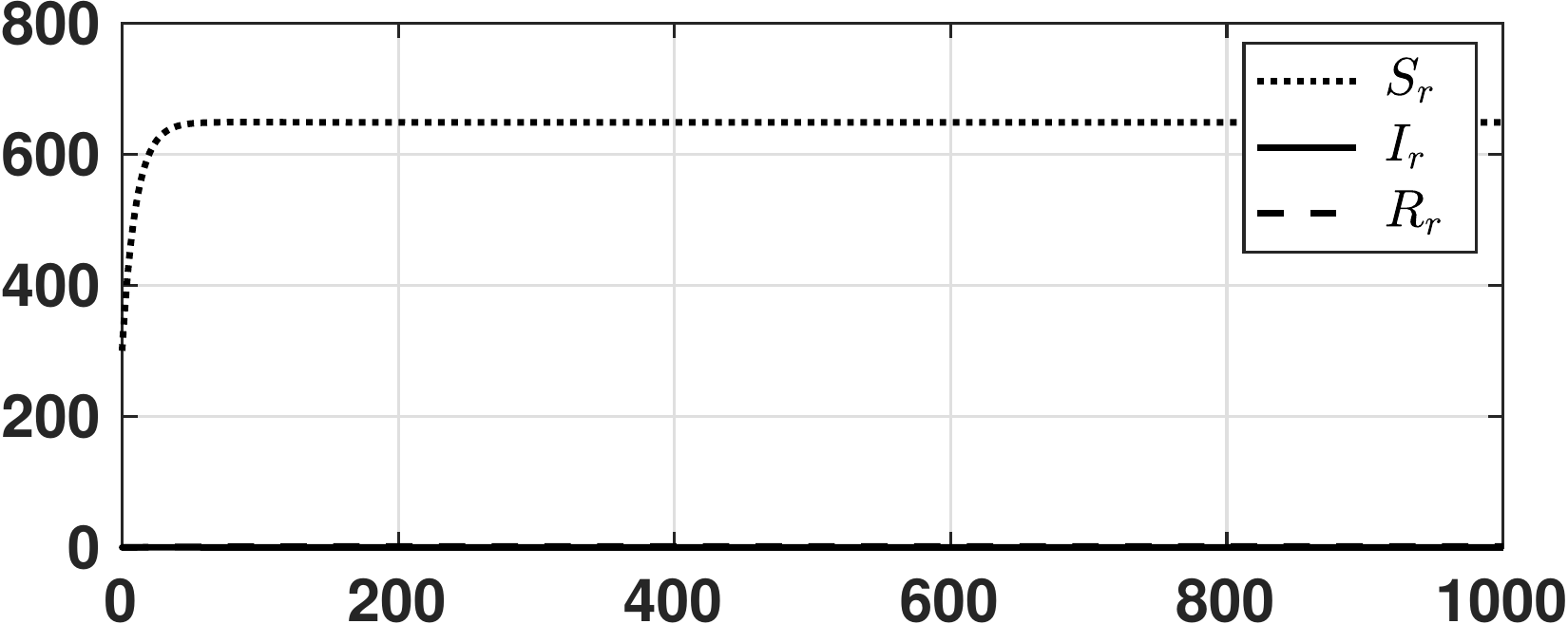}}
\caption{Illustration of the influence of movement on disease dynamics in the urban and rural patches.
(Top) In the absence of movement (i.e., $\dur=\dru = 0$), the urban patch reaches an endemic state. (Bottom) For $\dur=\dru=0.05$, both populations reach a disease-free steady state. In this sense, movement allows the disease to die out; see Example \ref{ex:2}.
 \label{fig:ex2a}}
\end{figure}

\begin{figure}[htb!]
\subfloat[Total number of infected individuals for $I_{0u} = 100$ and $\bu = 3\cdot 10^{-2}$. \label{fig:ex2c1}]{\includegraphics[width=0.49\linewidth]{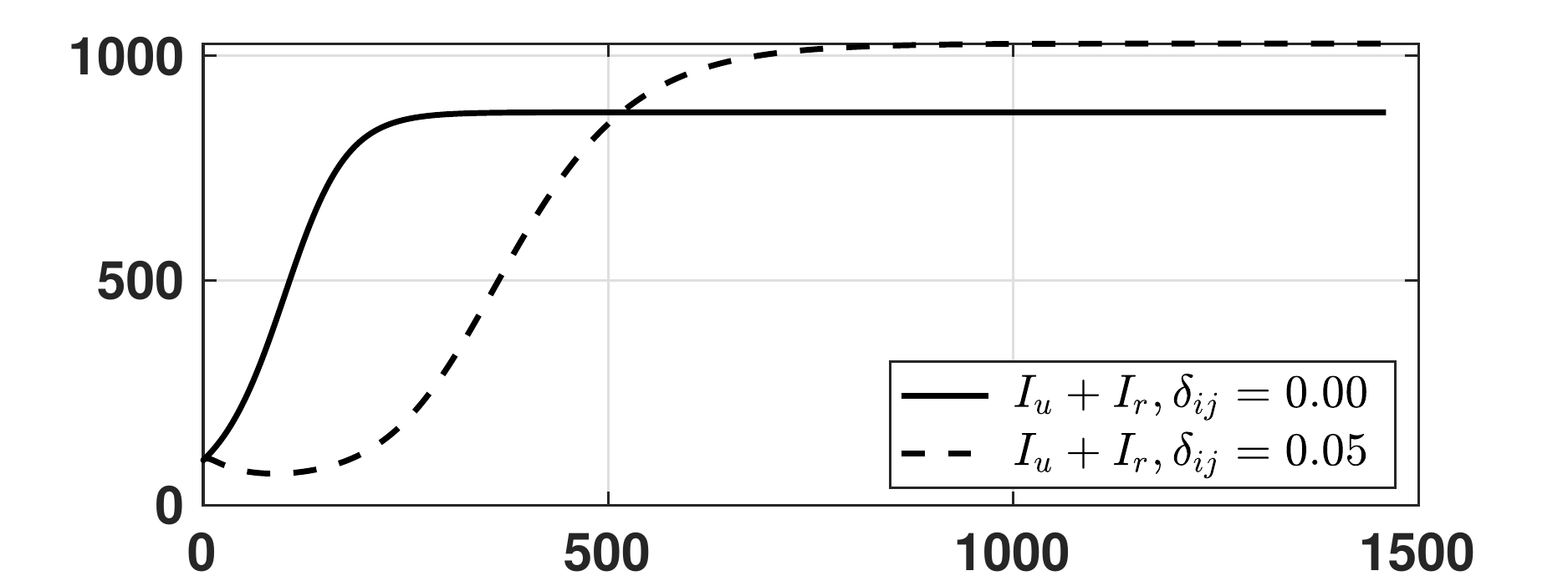}}
\hfill
\subfloat[Total number of infected individuals for $I_{0u} = 1$ and $\bu = 5.3\cdot 10^{-2}$ \label{fig:ex2c2}]{\includegraphics[width=0.49\linewidth]{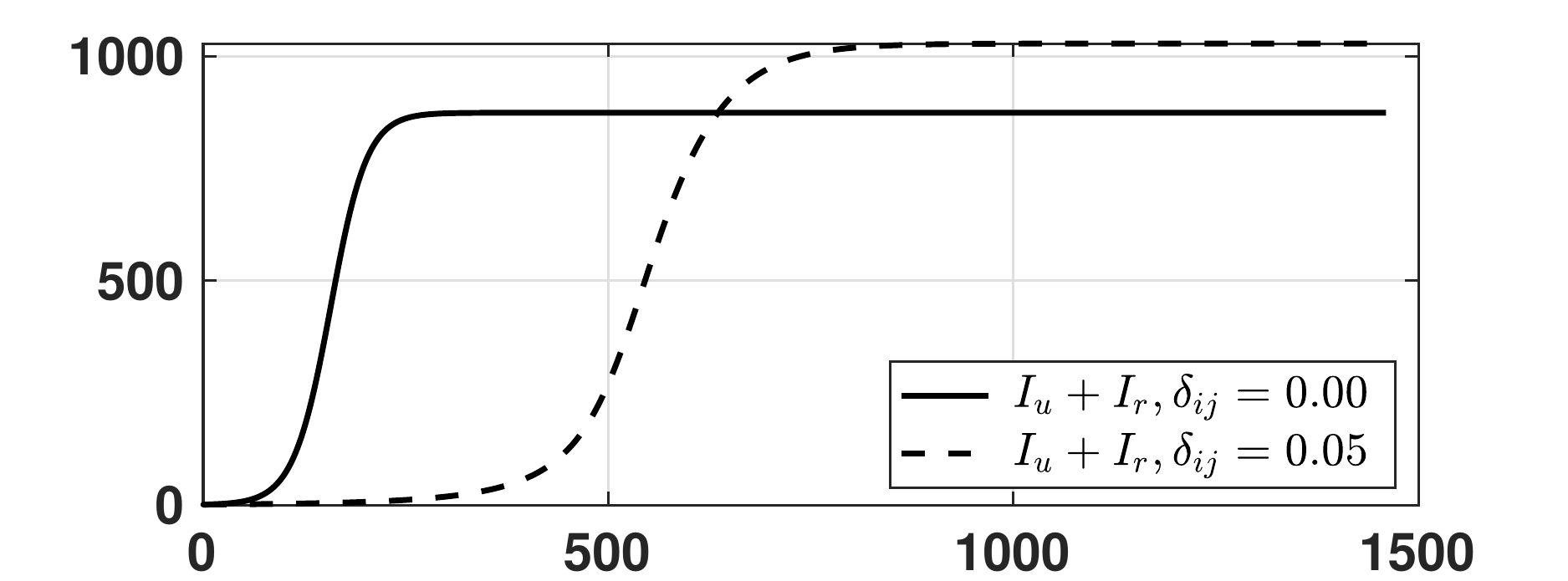}}
\caption{The efect of conditions \eqref{eq:condR0} on infected classes of individuals; see Example \ref{ex:2}. 
\label{fig:ex2c}}
\end{figure}
}
\end{example}

\begin{example} \label{ex:3} % file ex2b.m
{\rm
We now give an example with two endemic states. We use the parameter values
\begin{equation*}
	\begin{array}{ccccc}
\mm = 1/(365\cdot 70)\,, & \pu = 0.80\,, & \gu = 0.15\,, & \bu = 3\cdot 10^{-4}\,, & \dur = 0.04\,,\\
\mr = 1/(365\cdot 70)\,, & \pr = 0.40\,, & \gr = 0.10\,, & \br = 2\cdot 10^{-5}\,, & \dru = 0.05\,.
	\end{array}
\end{equation*}
We fix the initial conditions 
\begin{equation*}
\begin{array}{ccc}
	N_{u0} = 1000\,, & \widetilde{S}_{u0} = 0\,,\\
N_{r0} = 300\,,  & I_{r0} = 0 \,, & \widetilde{S}_{r0} = 0
\end{array}
\end{equation*}
and vary the initial number of infected individuals in the urban patch from $0$ to $N_u$. In this case, there are two endemic states, with $(I_u^*/N_u^*, I_r^*/N_r^*) \approx (0.09, 0.07)$ and $(I_u^*/N_u^*, I_r^*/N_r^*) \approx (0.49, 0.45)$, which we illustrate in Figure \ref{fig:remarkEnd}. We show $I_u^*$ and $I_r^*$ as a function of $I_{0u}$ in Figure \ref{fig:ex3a}. We observe that the steady state $(I_u^*/N_u^*, I_r^*/N_r^*) \approx (0.09, 0.07)$  is unstable.

\begin{figure}[htb!]
\centering%no!\hfill
\includegraphics[width=.7\linewidth]{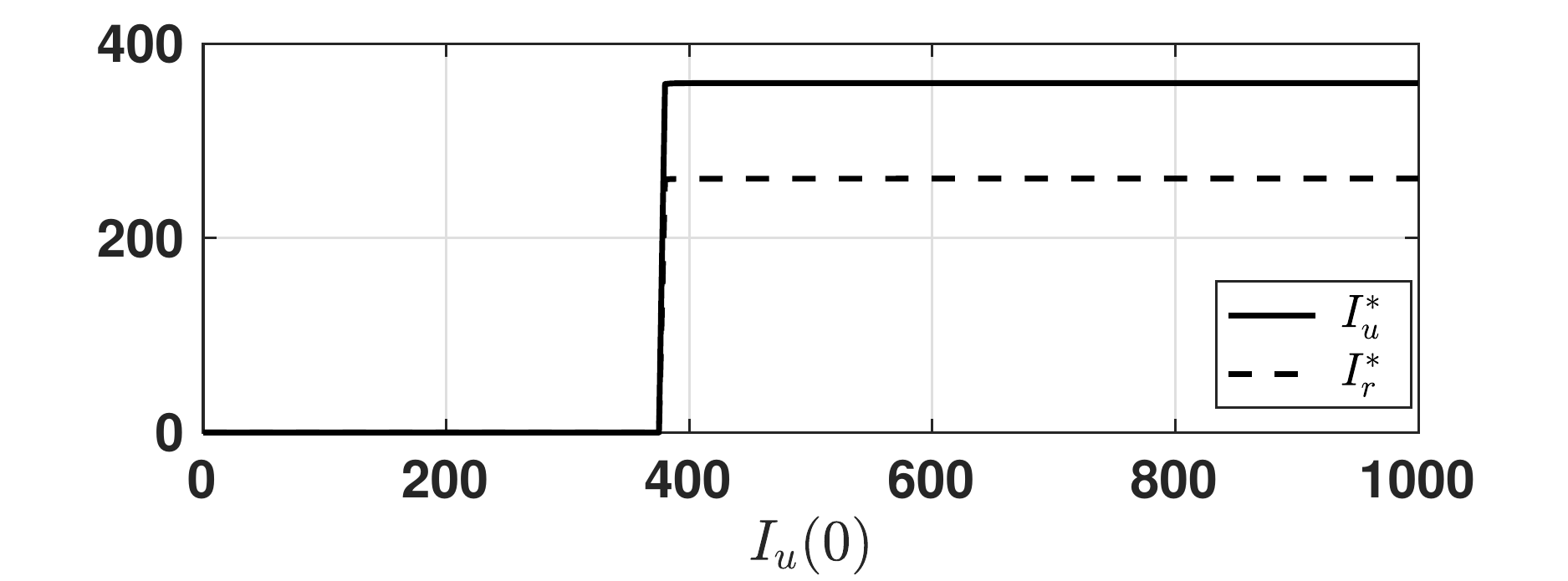}
\caption{The number of infected individuals in the urban and rural patches in an endemic steady state. We observe convergence of $I_u^*$ to a persistent number of infections as a function of $I_{0u}$; see Example \ref{ex:3}. 
\label{fig:ex3a}}
\end{figure}
}
\end{example}

\begin{example} \label{ex:6}% file ex6.m
{\rm
In this example, we explore the dependence on $\mathcal{R}_{0u}$ and $\mathcal{R}_{0r}$ in our model to illustrate the region of local asymptotic stability from Eqs.~\eqref{eq:condR0}. We consider the parameter values
\begin{equation*}
	\begin{array}{ccccc}
\mm = 1/(365\cdot 70)\,, & \pu = 0.80\,, & \gu = 0.01\,, & \dur = 0.005\,,\\
\mr = 1/(365\cdot 70)\,, & \pr = 0.40\,, & \gr = 0.05\,, & \dru = 0.010\,.
	\end{array}
\end{equation*}
We fix the initial conditions 
\begin{equation*}
	\begin{array}{ccc}
N_{u0} = 999\,, & I_{u0} = 1 \,,&\widetilde{S}_{u0} = 0\,,\\
N_{r0} = 300\,,  & I_{r0} = 0 \,, &\widetilde{S}_{r0} = 0
	\end{array}
\end{equation*}
and vary $\bu$ and $\br$; see our results in Figure \ref{fig:ex6}. We observe that the conditions in Eqs.~\eqref{eq:condR0} precisely describe the disease-free region (in gray). Therefore, $\mathcal{R}_{0u}$ or $\mathcal{R}_{0r}$ can have a value that is slightly larger than $1$ in situations with a disease-free steady state. This is not the case when we consider just one patch, so movement can be beneficial.

\begin{figure}[htb!]
\centering
{\includegraphics[width=0.65\linewidth]{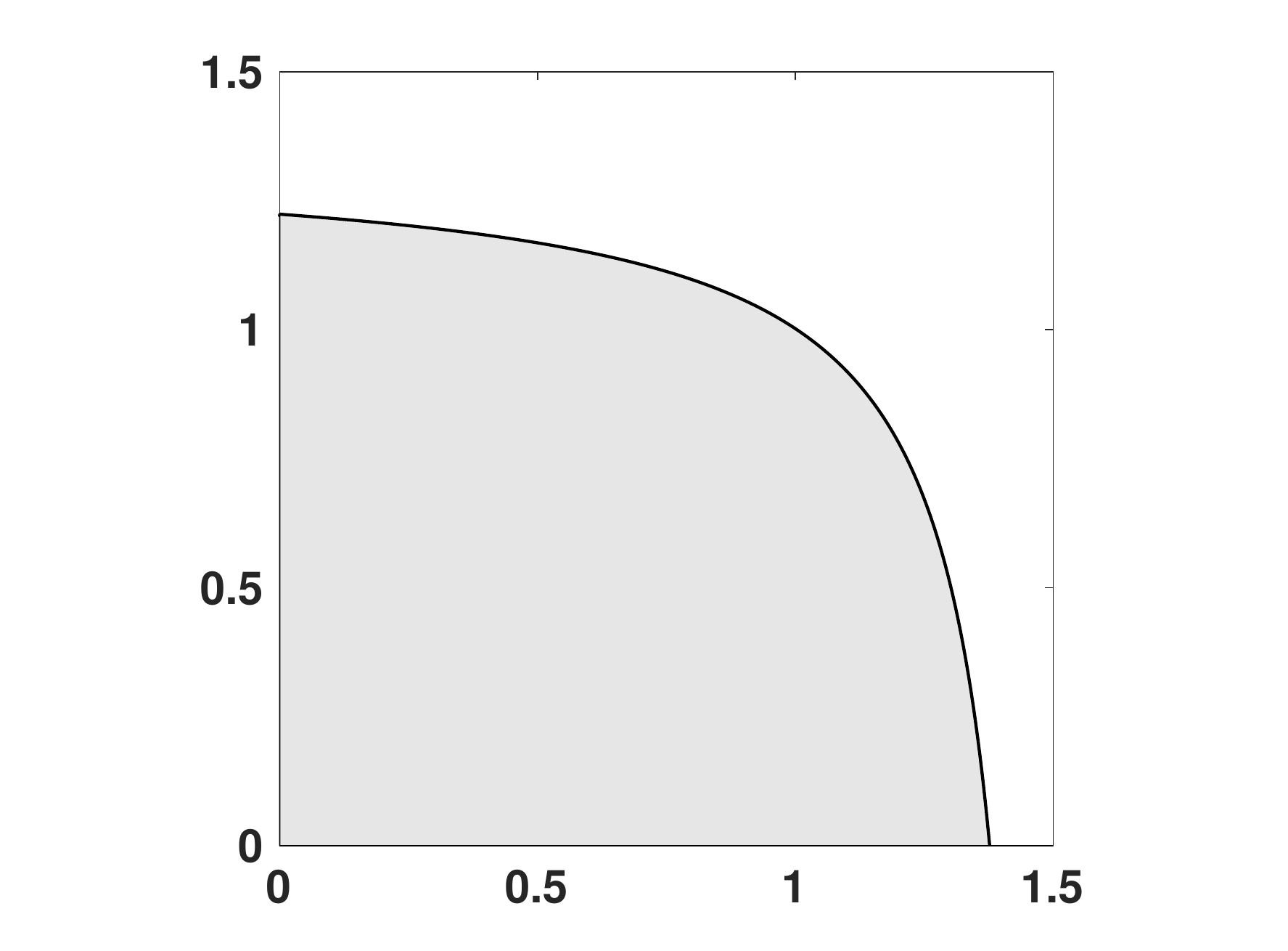}}
\caption{Population as a function of the local basic reproduction numbers in the urban $\mathcal{R}_{0u}$ (horizontal axis) and rural $\mathcal{R}_{0r}$ (vertical axis) environments; see Example \ref{ex:6}. The gray area represents the region of local asymptotic stability of both (urban and rural) populations from Lemma \ref{lem:diseaseFree}.
\label{fig:ex6}}
\end{figure}
}

\end{example}

\begin{example} \label{ex:6b} 
{\rm
In this example, we study how a disease spreads through the two populations as we vary $\dur$ and $\dru$ when initially one infected person arrives at one patch. We consider the parameter values
\begin{equation*}
	\begin{array}{ccccc}
\mm = 1/(365\cdot 70)\,, & \pu = 0.80\,, & \gu = 0.15\,, & \bu = 3\cdot{10}^{-1}\,,\\
\mr = 1/(365\cdot 70)\,, & \pr = 0.40\,, & \gr = 0.10\,, & \br = 2\cdot{10}^{-3}\,.
	\end{array}
\end{equation*}
We fix the initial conditions 
\begin{equation*}
	\begin{array}{ccc}
N_{u0} = 99999\,, & I_{u0} = 1\,, &\widetilde{S}_{u0} = 0\,,\\
N_{r0} = 30000\,,  & I_{r0} = 0\,, &\widetilde{S}_{r0} = 0
	\end{array}
\end{equation*}
and vary $\dur$ and $\dru$. We show our results in Figure \ref{fig:ex6B}.

\begin{figure}[htb!]
\centering
\subfloat[Steady-state fraction ($I_u^*$) of infected individuals in the urban patch.  \label{fig:ex6A}]{\includegraphics[width=0.49\linewidth]{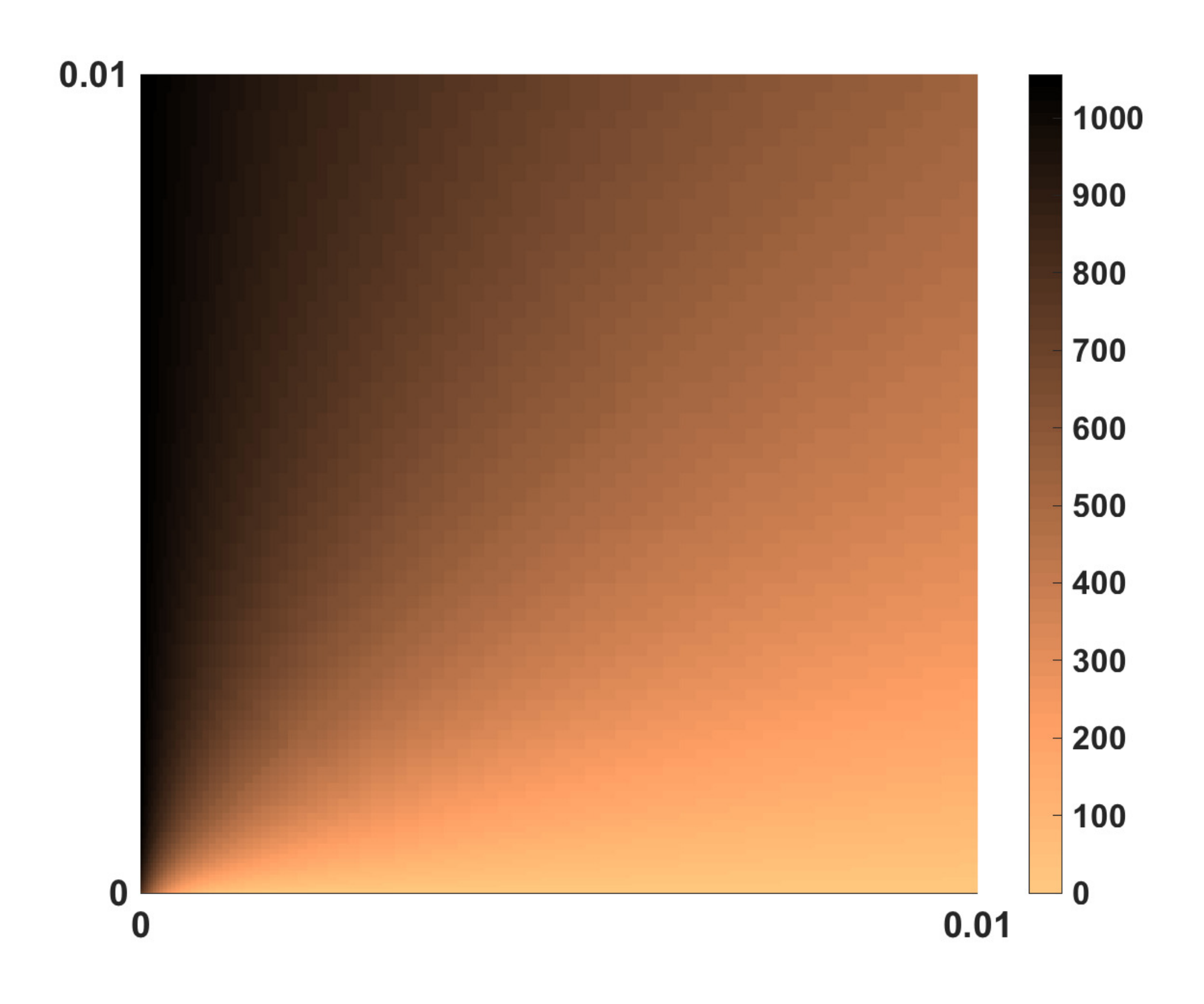}}
\hfill
\subfloat[Steady-state fraction ($I_r^*$) of infected individuals in the rural patch. \label{fig:ex6B_}]{\includegraphics[width=0.49\linewidth]{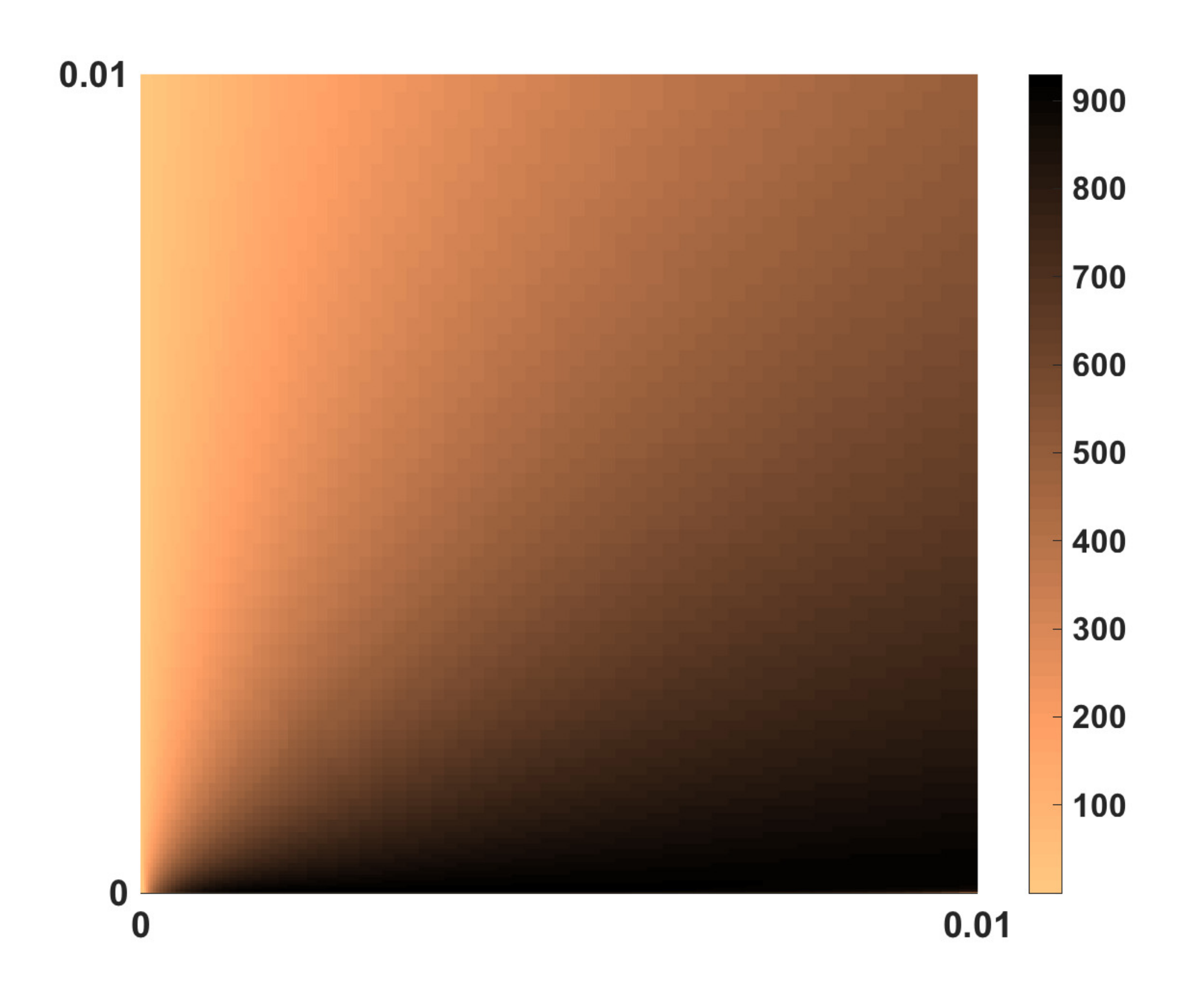}}
\caption{The effect of a single infected individual in the urban patch as a function of movement $\dur$ (horizontal axis) and $\dru$ (vertical axis); see Example \ref{ex:6b}.
\label{fig:ex6B}}
\end{figure}
}
\end{example}

\begin{example} \label{ex:varDeltas}% file ex3_delta_var.m
{\rm
To model different rates of movement between patches on weekdays and weekends, we now take $\dur$ and $\dru$ to be piecewise constant and periodic. Specifically, we take $\dur(t)$ and $\dru(t)$ as in Figure \ref{fig:deltas}. In Figure \ref{fig:ex3_delta_var}, we show the results of our numerical computations using the parameter values 
\begin{equation*}
	\begin{array}{ccccc}
\mm = 1/(365\cdot 70)\,, & \pu = 0.80\,, & \gu = 0.15\,, & \bu = 3\cdot{10}^{-4}\,,\\
\mr = 1/(365\cdot 70)\,, & \pr = 0.40\,, & \gr = 0.10\,, & \br = 2\cdot{10}^{-5}
	\end{array}
\end{equation*}
with initial conditions 
\begin{equation*}
	\begin{array}{ccc}
N_{u0} = 999\,, & I_{u0} = 1\,, &\widetilde{S}_{u0} = 0\,,\\
N_{r0} = 300\,,  & I_{r0} = 0\,, &\widetilde{S}_{r0} = 0\,.
	\end{array}
\end{equation*}

\begin{figure}[ht!]
\centering
\subfloat{\includegraphics[width=0.5\linewidth]{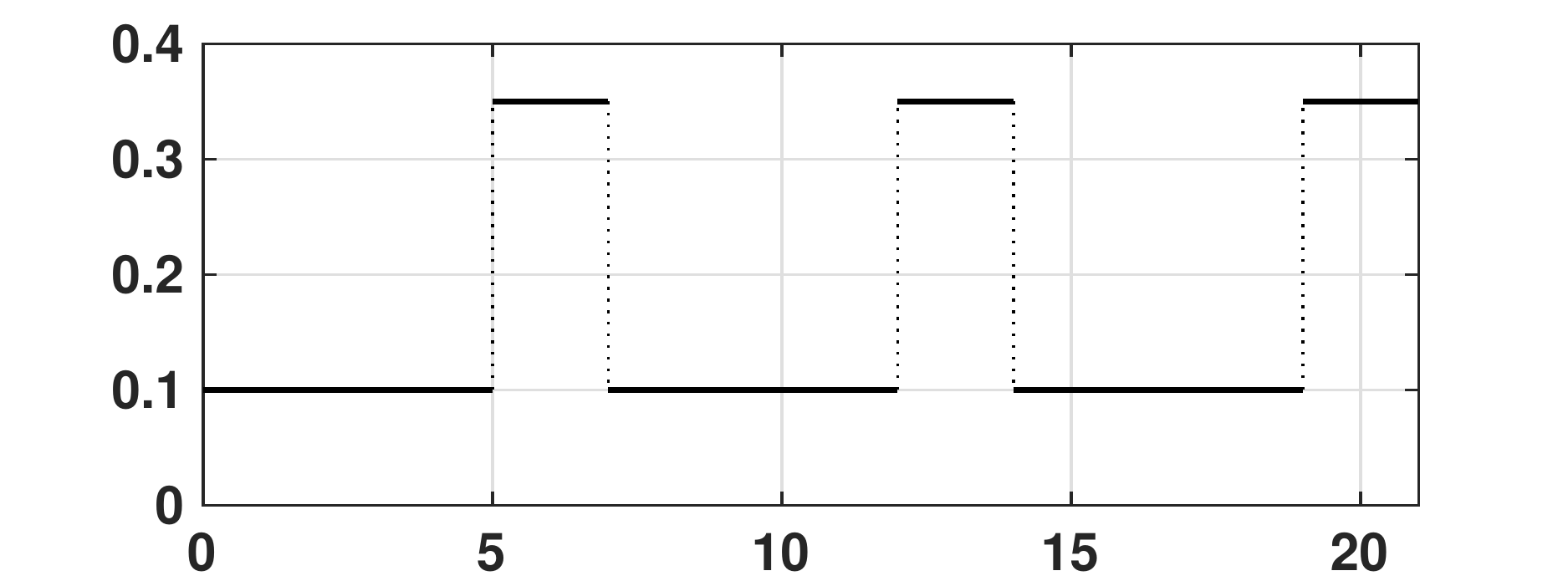}}
\hfill
\subfloat{\includegraphics[width=0.5\linewidth]{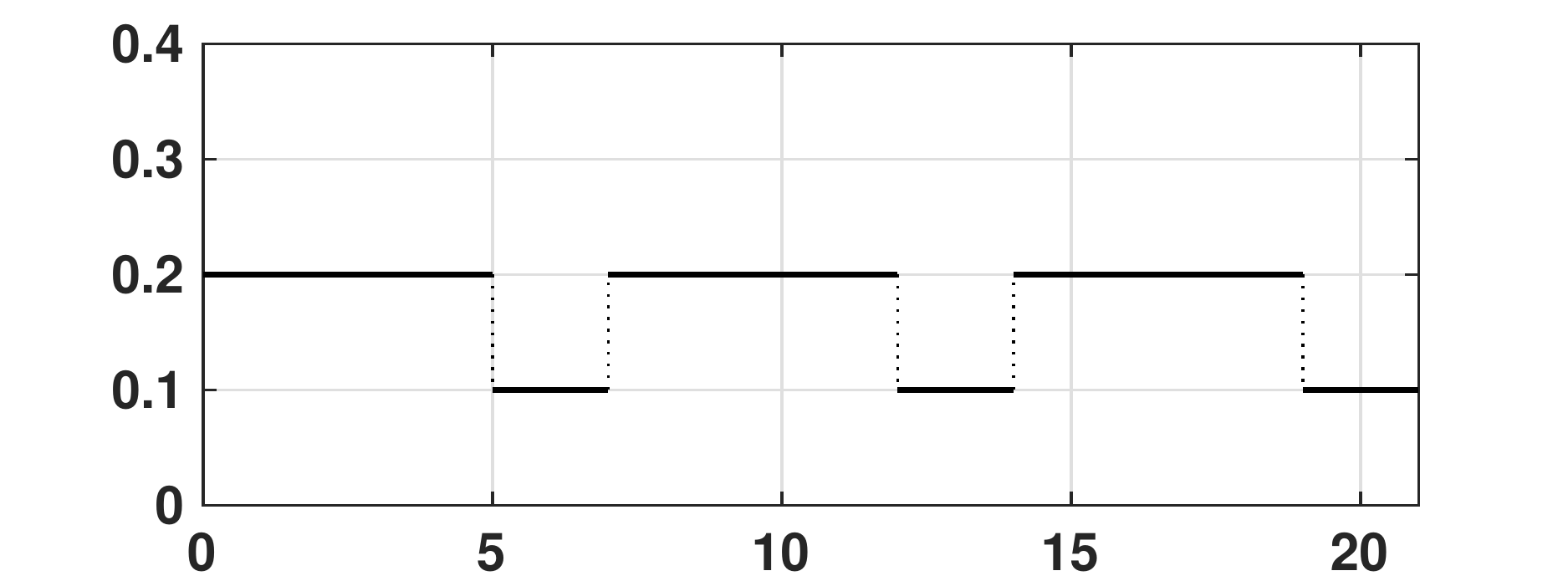}}
\caption{Periodic, piecewise constant values of the movement (left) $\dur$ and (right) $\dru$ as a function of time (in days). We use these functions in Example \ref{ex:varDeltas}.
\label{fig:deltas}}
\end{figure}

\begin{figure}[ht!]
\centering
\subfloat{\includegraphics[width=0.49\linewidth]{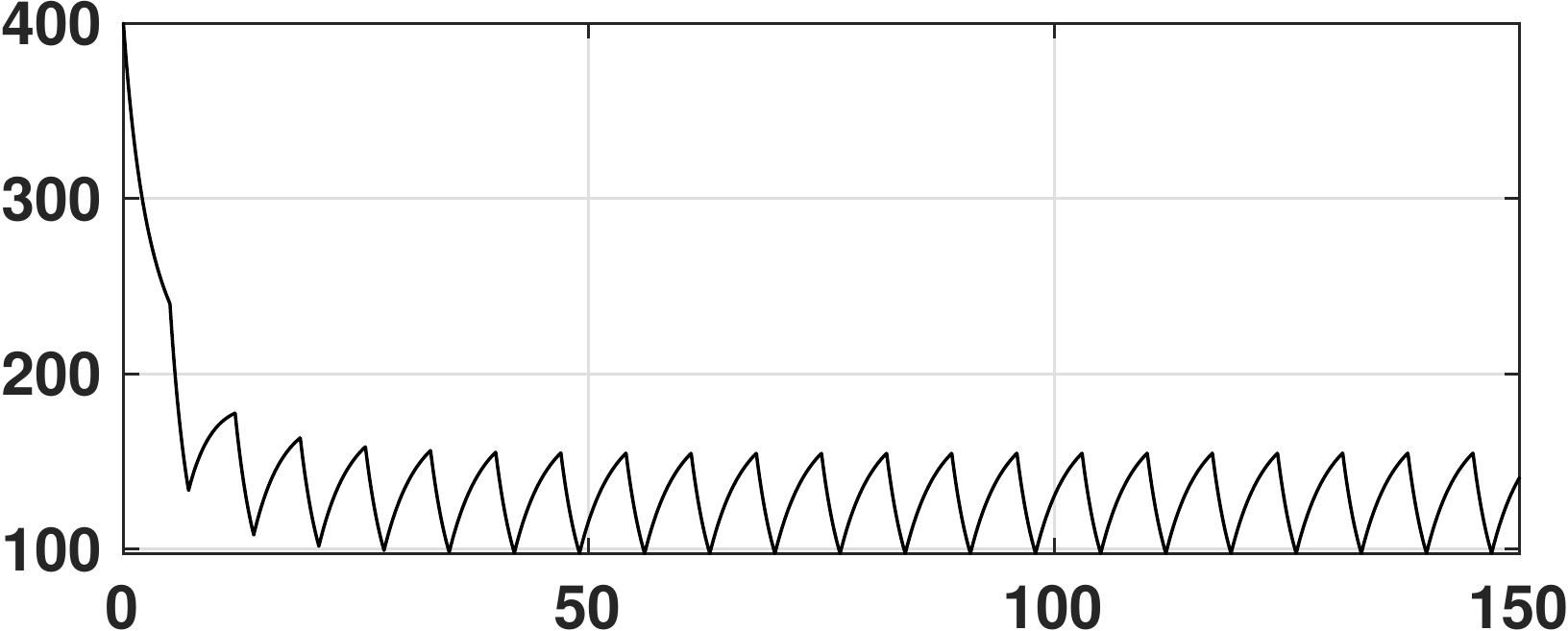}}
\hfill
\subfloat{\includegraphics[width=0.49\linewidth]{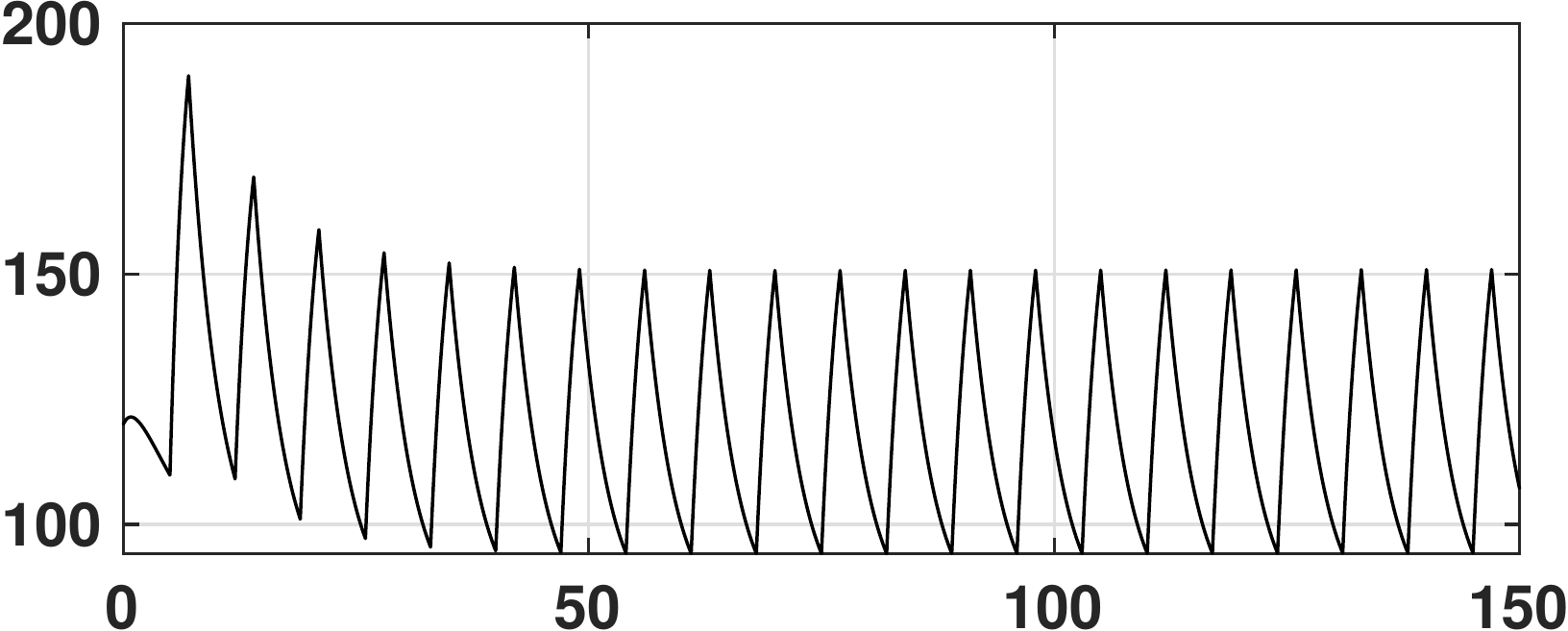}}
\caption{The number of infected individuals in the (left) urban (i.e., $I_u(t)$) and (right) rural $I_r(t)$ populations as a function of time (in days) for the periodic, piecewise constant movements $\dur$ and $\dru$ from Figure \ref{fig:deltas}; see Example \ref{ex:varDeltas}.
\label{fig:ex3_delta_var}}
\end{figure}
}
\end{example}

\begin{example} \label{eq:ninePts} {\rm 

We now revisit Remark \ref{rem:severalPoints}, where we stated that we can observe numerically the existence of several distinct endemic steady states. Equation \eqref{eq_systemG1G2} is equivalent to a polynomial equation (as a function of $B_u$ or $B_r)$ of degree $9$, for which there can exist $9$ solutions (for $i_u^*$ or $i_r^*$). For the parameter values
\begin{equation*}
	\begin{array}{ccccc}
\mm = 1/(365\cdot 70)\,, & \pu = 1\,, & \gu = 0.15\,, & \bu = 3\cdot{10}^{-4}\,, & \dur = 4\cdot 10^{-6}\,,\\
\mr = 1/(365\cdot 70)\,, & \pr = 2\,, & \gr = 1.00\,, & \br = 2\cdot{10}^{-3}\,, & \dru = 5\cdot 10^{-5}\,,
	\end{array}
\end{equation*}
we see that all $9$ solutions are in the interval $[0,1]$. They correspond to feasible steady states $(i_u^*,i_r^*)$; see Figure \ref{fig:ex_9pts}. We see numerically that $4$ of these points are locally stable.

\begin{figure}[ht!]
\centering
{\includegraphics[width=0.7\linewidth]{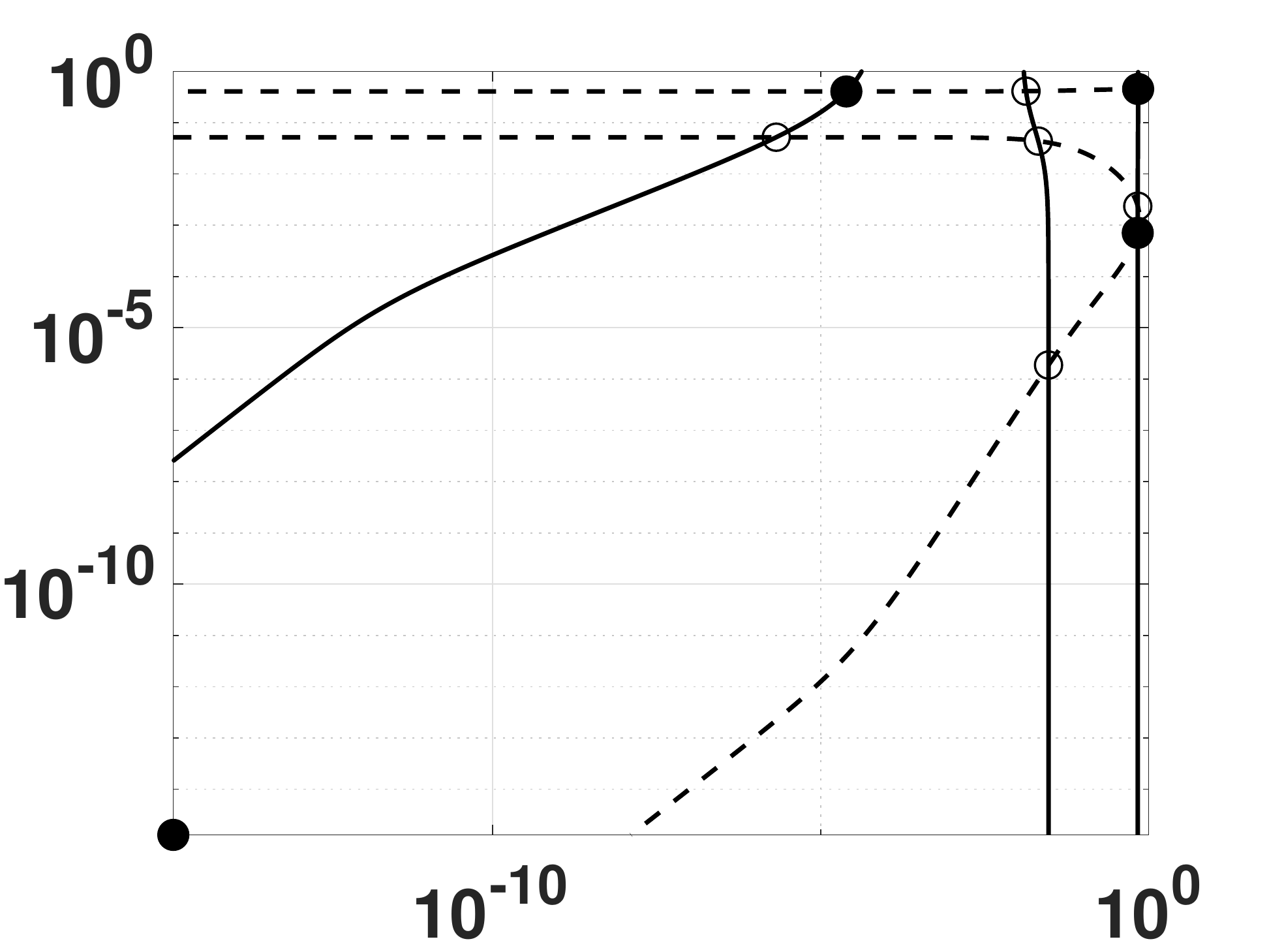}}
\caption{We plot $G_1(i_u,i_r)=0$ versus $G_2(i_u,i_r)=0$. We observe numerically that there are $4$ locally stable steady states, illustrating that movement between patches can introduce multiple steady states. We show the locally stable steady states with solid disks and the other steady states with open disks.
\label{fig:ex_9pts}} 
\end{figure}

}
\end{example}

%%%%

\section{Conclusions and Discussion} \label{sec:disc}

The dynamics of spreading diseases are influenced significantly by spatial heterogeneity and population movement. In this paper, we illustrated the importance of incorporating movement into models of disease dynamics using a simple but biologically meaningful model. Specifically, we constructed a two-patch compartmental model that incorporates movement between urban and rural populations, as well as the possibility of reinfection after recovery. 

When there is a lot of population movement, there are regions of local asymptotic stability of the disease-free steady-state even when the basic reproduction number $\mathcal{R}>1$. This arises predominantly from the numerous individuals who move between patches. 

The exploration of interacting populations plays an important role in the understanding of disease dynamics. Many models of disease spreading focus on a single population \cite{ccc}, but populations do not exist in isolation. Using our two-patch model with urban and rural environments, we illustrated several examples of plausible real-world scenarios in which movement yields insightful information about disease spreading and epidemics. We expect that such dynamics will be relevant for studies of disease spreading on networks, such as when many people commute daily between their homes in rural areas and work in urban centers (as is the case in many countries in South and Central America).

%%%%%

\section{Acknowledgements}

We thank the Research Center in Pure and Applied Mathematics and the Mathematics Department at Universidad de Costa Rica for their support during the preparation of this manuscript. The authors gratefully acknowledge institutional support for project B8747 from an UCREA grant from the Vice Rectory for Research at Universidad de Costa Rica. We also acknowledge helpful discussions with Prof. Luis Barboza, Prof. Carlos Castillo-Chavez, and Prof. Esteban Segura.

%%%%%%

%\newpage

\end{document}